\documentclass[12pt]{article}
\usepackage{
amsmath,amssymb,mathrsfs,bbm,a4wide}

\usepackage{tikz}
\usetikzlibrary{matrix,arrows}

\newenvironment{proofof}[1]{%
\par\addvspace{12pt plus3pt minus3pt}\global\logotrue%
\noindent{\bf Proof of #1.\hskip.5em}\ignorespaces}{%
	\par\iflogo\prbox\par
	\addvspace{12pt plus3pt minus3pt}\fi}

\newif\iflogo
\def\prbox{\par
	\vskip-\lastskip\vskip-\baselineskip\hbox to \hsize{\hfill\fboxsep0pt\fbox{\phantom{\vrule width5pt height5pt depth0pt}}}\global\logofalse}

\newenvironment{proof}{%
\par\addvspace{12pt plus3pt minus3pt}\global\logotrue%
\noindent{\bf Proof.\hskip.5em}\ignorespaces}{%
	\par\iflogo\prbox\par
	\addvspace{12pt plus3pt minus3pt}\fi}

\usepackage{amsmath,amssymb,amsfonts}
\newcommand{\bbLbrack}{[\kern-0.4em{[}\,}
\newcommand{\bbRbrack}{\,]\kern-0.4em{]}}
\usepackage[T1]{fontenc}
\usepackage{yfonts}

\usepackage{extarrows}
\usepackage{undertilde}

\usepackage{mathrsfs,latexsym}

\newcommand{\II}{{\boldsymbol{1}}}

\newcommand{\CC}{{\mathbb C}}

\newcommand{\RR}{{\mathbb R}}

\newcommand{\NN}{{\mathbb N}}

\newcommand{\ZZ}{{\mathbb Z}}

\newcommand{\CoinX}[1]{C_0^\infty({#1})}

\newtheorem{Thm}{Theorem}[section]

\newtheorem{Def}[Thm]{Definition}
\newtheorem{Lem}[Thm]{Lemma}
\newtheorem{Prop}[Thm]{Proposition}
\newtheorem{Cor}[Thm]{Corollary}

\newcommand{\Js}{{\mathsf J}}

\newcommand{\DD}{{\mathscr D}}

\newcommand{\HH}{{\mathscr H}}

\newcommand{\KK}{{\mathscr K}}
\newcommand{\FF}{{\mathscr F}}

\newcommand{\Ac}{{\mathcal A}}
\newcommand{\Bc}{{\mathcal B}}
\newcommand{\Cc}{{\mathcal C}}
\newcommand{\Dc}{{\mathcal D}}
\newcommand{\Fc}{{\mathcal F}}
\newcommand{\Hc}{{\mathcal H}}
\newcommand{\cN}{{\mathcal N}}

\newcommand{\OO}{{\mathscr O}}

\newcommand{\fb}{{\boldsymbol{f}}}
\newcommand{\gb}{{\boldsymbol{g}}}
\newcommand{\hb}{{\boldsymbol{h}}}

\newcommand{\xb}{{\boldsymbol{x}}}

\newcommand{\Tb}{{\boldsymbol{T}}}

\newcommand{\Bfr}{{\mathfrak B}}
\newcommand{\Ffr}{{\mathfrak F}}
\newcommand{\Mfr}{{\mathfrak M}}

\newcommand{\ogth}{{\mathfrak o}}
\newcommand{\tgth}{{\mathfrak t}}

\newcommand{\etap}{\eta'}

\newcommand{\supp}{{\rm supp}\,}
\newcommand{\Span}{{\rm span}\,}

\renewcommand{\Im}{{\rm Im}\,}

\newcommand{\End}{{\rm End}}
\newcommand{\Sym}{{\rm Sym}}

\newcommand{\dvol}{d\textrm{vol}}

\newcommand{\ip}[2]{{\langle #1\mid #2\rangle}}

\newcommand{\Ob}{{\boldsymbol{0}}}

\newcommand{\Cb}{{\boldsymbol{C}}}
\newcommand{\Db}{{\boldsymbol{D}}}

\newcommand{\Lb}{{\boldsymbol{L}}}
\newcommand{\Mb}{{\boldsymbol{M}}}
\newcommand{\Nb}{{\boldsymbol{N}}}

\newcommand{\Mc}{{\mathcal{M}}}
\newcommand{\Nc}{{\mathcal{N}}}
\newcommand{\Sc}{{\mathcal{S}}}

\newcommand{\Ct}{{\sf C}}

\newcommand{\LCT}{{\sf LCT}}
\newcommand{\Loco}{{\sf Loc}_0}
\newcommand{\Loc}{{\sf Loc}}
\newcommand{\preSympl}{{\sf preSympl}}
\newcommand{\Set}{{\sf Set}}

\newcommand{\Sympl}{{\sf Sympl}}

\newcommand{\Alg}{{\sf Alg}}
\newcommand{\CAlg}{{\sf C^*\hbox{-}Alg}}

\newcommand{\Phys}{{\sf Phys}}

\newcommand{\grAS}{{\sf grAS}}
\newcommand{\grCAS}{{\sf grC^*AS}}

\newcommand{\Af}{{\mathscr A}}

\newcommand{\Bf}{{\mathscr B}}

\newcommand{\Cf}{{\mathscr C}}
\newcommand{\Df}{{\mathscr D}}

\newcommand{\Ff}{{\mathscr F}}

\newcommand{\Qf}{{\mathscr Q}}

\newcommand{\Sf}{{\mathscr S}}
\newcommand{\Tf}{{\mathscr T}}
\newcommand{\Uf}{{\mathscr U}}
\newcommand{\Vf}{{\mathscr V}}
\newcommand{\Tc}{{\mathcal T}}

\newcommand{\Wf}{{\mathscr W}}
\newcommand{\Xf}{{\mathscr X}}
\newcommand{\Yf}{{\mathscr Y}}

\newcommand{\Sol}{{\mathscr L}}

\newcommand{\id}{{\rm id}}

\newcommand{\nto}{\stackrel{.}{\to}}

\newcommand{\Aut}{{\rm Aut}}

\newcommand{\Fld}{{\rm Fld}}

\DeclareMathOperator{\Fol}{Fol}

\DeclareMathOperator{\cl}{cl}

\DeclareMathOperator{\co}{co}
\DeclareMathOperator{\eq}{eq}

\newcommand{\rce}{{\rm rce}}

\newcommand{\kin}{{\rm kin}}

\newcommand{\cPhi}{\check{\Phi}}

\DeclareMathOperator{\CCR}{CCR}

\newcounter{tightenum}
\newenvironment{tightitemize}%
{\begin{list}{$\bullet$}{\setlength{\itemsep}{0pt}\setlength{\parsep}{0pt}\setlength{\topsep}{0pt}}}%
{\end{list}}
\newenvironment{tightenumerate}%
{\begin{list}{(\roman{tightenum})}{\usecounter{tightenum} \setlength{\itemsep}{0pt}\setlength{\parsep}{0pt}\setlength{\topsep}{0pt}}}%
{\end{list}}

\newcounter{assumptions}

\newcommand{\cPsi}{\check{\Psi}}
\newcommand{\tS}{\widetilde{S}}
\newcommand{\elb}{{\boldsymbol{\ell}}}

\begin{document}


\title{Endomorphisms and automorphisms of locally covariant quantum field theories}
\author{Christopher J Fewster\thanks{\tt chris.fewster@york.ac.uk}\\ Department of Mathematics,
                 University of York, \\
                 Heslington,
                 York YO10 5DD, U.K.}
\date{\today}

\maketitle

\begin{abstract}
In the framework of locally covariant quantum field theory, a theory is described as a
functor from a category of spacetimes to a category of $*$-algebras. It is 
proposed that the global gauge group of such a theory can be identified as the
group of automorphisms of the defining functor. Consequently, multiplets of fields may be
identified at the functorial level. It is shown that locally covariant theories that obey
standard assumptions in Minkowski space, including energy compactness,
have no proper endomorphisms (i.e., all endomorphisms are automorphisms) and
have a compact automorphism group.
Further,  it is shown how the endomorphisms and automorphisms of a locally covariant theory may, in principle, be classified in any single spacetime. As an example, the endomorphisms and automorphisms of a system of finitely many free scalar fields are completely classified. 
\end{abstract}

\renewcommand{\thefootnote}{\alph{footnote}}

\section{Introduction}

Algebraic quantum field theory (AQFT) has been highly successful in 
analysing the structural properties of general quantum field theories
in Minkowski space~\cite{Haag}. For many years, however, rigorous quantum field
theory in curved spacetimes was restricted to particular free models, 
or to spacetimes of maximal symmetry. This situation has changed,
following the introduction, by Brunetti, Fredenhagen and Verch (henceforth abbreviated as BFV) of a framework of locally covariant quantum field
theory~\cite{BrFrVe03}. 
This framework, in which a quantum field theory is defined
as a functor between a category of spacetimes and a category of
$(C)^*$-algebras,
developed from a formulation given in Verch's proof
of a general spin--statistics connection~\cite{Verch01} and
has subsequently played an important role in the completion of the perturbative 
construction of interacting theories in curved spacetime 
\cite{BrFr2000,Ho&Wa01,Ho&Wa02}, a Reeh--Schlieder theorem~\cite{Sanders_ReehSchlieder},
and an analysis of superselection sectors~\cite{Br&Ru05,BrunettiRuzzi_topsect}. It has also proved
possible to begin an analysis of the fundamental question of what it is
that makes a theory of physics the same in different spacetimes~\cite{FewVer:dynloc_theory, FewVer:dynloc2} (see also~\cite{FewsterRegensburg} for a short summary).  In addition, more quantitative applications have been made to particular models in the context
of the Casimir effect and quantum energy inequalities~\cite{Few&Pfen06,Fewster2007,FewOlumPfen}
and to cosmology~\cite{DapFrePin2008,DegVer2010,VerchRegensburg}. 

One of the major successes of AQFT in flat spacetime is undoubtedly the
Doplicher--Haag--Roberts (DHR) analysis of superselection sectors and the reconstruction of the field algebra  and gauge group from the algebra of observables in the vacuum sector~\cite{DHRi,DHRii,DopRob:1990}.
Brunetti and Ruzzi~\cite{Br&Ru05} have employed ideas
of local covariance in order to develop a parallel analysis in curved spacetime. One could characterize
their approach as being `bottom-up', as the analysis is performed in each spacetime
and questions of covariance are then addressed. From the functorial viewpoint,
it would be more natural to proceed in a `top-down' manner, identifying the 
relevant structures and definitions as properties of the functor defining
the theory. With this eventual aim in mind, the present paper begins by discussing how the 
global gauge group of a locally covariant theory may be identified in terms of
the functor. Our proposal is simply that the global gauge group is
the group of automorphisms of the functor defining the theory. 

To explain this, we recall, first, that every functor $\Af$ has an associated
group of automorphisms $\Aut(\Af)$, whose elements are the natural
isomorphisms from $\Af$ to itself, equipped with the group structure induced
by composition.\footnote{We refer
the reader to~\cite{MacLane} for the main ideas of category theory
and to~\cite{BrFrVe03,FewVer:dynloc_theory, FewVer:dynloc2} for
their application in QFT. A brief exposition appears in Sec.~\ref{sect:framework}.}
The automorphism group often carries important structural
information about the functor and so, even from a purely mathematical
perspective, it is important to understand
what significance can be assigned to the automorphism group of a functor $\Af$
defining a locally covariant QFT, or, more generally, to the monoid of endomorphisms
$\End(\Af)$ of $\Af$.

The physical motivation for our study rests on the interpretation
of a natural transformation  $\zeta:\Af\nto\Bf$ between
functors describing two theories $\Af$ and $\Bf$ as an 
embedding of theory $\Af$ as a subtheory of $\Bf$~\cite{BrFrVe03,FewVer:dynloc_theory}. 
Thus an endomorphism
of $\Af$ is a way of embedding $\Af$ as a subtheory of itself, and an automorphism
is a means of doing this isomorphically. It is therefore very natural
to interpret the automorphism group as the global gauge group of the 
theory.\footnote{\label{fn:subgroup}There are circumstances where one might wish only to adopt
only a subgroup of the automorphisms as the group of internal symmetries; for example, the free complex scalar field is equivalent to two real scalars, but one would adopt the $U(1)$ subgroup of the full
$O(2)$ automorphism group in the former case. Again, it may be that 
the theory is studied as an approximation to a more detailed theory. 
For example, the Lagrangian of electromagnetism is invariant under negation
of the vector potential, but this symmetry is typically broken if one includes
interactions with matter.}

As we will show, this interpretation is supported by a comparison with the Minkowski space DHR analysis~\cite{DHRi}. 
There, the (maximal) global gauge group consists of (all) unitaries on the Hilbert space
of the vacuum representation of the field algebra that commute with the action of the
proper orthochronous Poincar\'e transformations, 
map each local field algebra isomorphically to itself and
preserve the vacuum state. As we show in Sec.~\ref{sec:gauge_gp}, these properties
are respected, and generalized, in our approach. To a large
extent they hold in a representation-independent sense in arbitrary spacetimes
(Sec.~\ref{sec:endos_autos}), 
and they hold in exactly the DHR sense in representations induced by gauge-invariant
states, with the Poincar\'e group replaced by the bijective spacetime isometries 
preserving (time)-orientation. In particular, (a) the automorphisms representing
global symmetries commute with spacetime symmetries -- thus, spacetime
and internal symmetries are completely independent, in a manner
reminiscent of the Coleman--Mandula theorem; (b) in Minkowksi space, the unitary implementation of the 
automorphism group in the Minkowski vacuum representation is a subgroup
of the maximal DHR group -- a key open question is to understand whether
these groups are equal. We also describe how the gauge group acts
on an abstract algebra of fields introduced in~\cite{Fewster2007} and allows the
definition of multiplets of fields at the functorial level. The fundamental
particle--antiparticle symmetry can be seen at this level. In addition, by taking
fixed-points under the action of the gauge group, one can define a new locally
covariant theory that is a candidate for the description of the observables of the
theory. This is not completely satisfactory, because it can happen that there are
no nontrivial fixed points (as occurs with the Weyl algebra if the gauge group is
continuous, for example); an approach that generates the observables in
suitable representations would be preferred, but is not pursued here. 

More generally, we will consider the proper endomorphisms of locally covariant theories, i.e., 
those embeddings of a theory as a subtheory of itself that are not gauge transformations.
For example, starting with any nontrivial theory $\Ff$, the countably infinite tensor product theory $\Ff^{\otimes\infty}: 
= \bigotimes_{n=1}^\infty \Ff$ admits a proper endomorphism with components $\zeta_\Mb A = \II_{\Ff(\Mb)}\otimes A$
(there are many others). The existence of proper endomorphisms seems, in general, to indicate pathological
behaviour: if $\eta:\Ff\nto\Ff$ is a proper endomorphism, so are all of its positive integer powers\footnote{As
$\eta$ is monic, it cannot happen that any two powers are equal, given that $\eta$ is not an isomorphism,
so these nested theories are indeed all distinct.}
and we obtain an infinite chain of properly nested subtheories of $\Ff$, each of which is itself equivalent
to $\Ff$. Any individual physical element of the theory, such as a species of particle, must
be replicated in each of these nested theories, suggesting that each such element
appears with infinite multiplicity. In Minkowski space QFT, it has long been understood~\cite{HaaSwi:1965,BucWic:1986}
that the latter situation is typically incompatible with a particle interpretation or good thermodynamic
properties. 

Accordingly, it is of interest to understand what conditions on a theory exclude the existence of proper endomorphisms. We address this issue in Sec.~\ref{sect:compactness} for theories described using $C^*$-algebras (to be 
thought of as field algebras) with a suitable state space and which obey a number of 
standard conditions in Minkowski space, most notably a mild energy compactness 
assumption, inspired by those of~\cite{HaaSwi:1965,BucWic:1986}. 
Under these assumptions, we show that any endomorphism
of the theory is unitarily implemented in the Poincar\'e invariant Minkowski representation
by an element of the maximal gauge global group in the
DHR sense~\cite{DHRi}. Therefore, assuming that there are
no `accidental' gauge symmetries in Minkowski space, every endomorphism is an
automorphism. We also give a direct proof that the automorphism group is compact;
see~\cite{DHRi,DopLon:1984} for proofs in Minkowski space under different hypotheses.

We emphasise that we are discussing endomorphisms of the functor, rather than
of the algebras corresponding to individual spacetimes, which admit many
proper endomorphisms (indeed the DHR superselection theory makes essential use of algebra endomorphisms). In any $\Mb$ spacetime with global topology $\RR^4$, for
example, the timeslice axiom gives an isomorphism $\Af(\Mb)\cong \Af(\Db)$, where $\Af(\Mb)$ is the algebra
associated to the full spacetime, while $\Af(\Db)$ is the algebra associated to the domain of dependence of a set that is diffeomorphic to 
an open $3$-ball and lies within a Cauchy surface. But as $\Af(\Db)$ 
is naturally embedded as a proper subalgebra of $\Af(\Mb)$, it follows that there is a proper endomorphism of $\Af(\Mb)$. 
Of course, this endomorphism is not at all canonical, because it depends sensitively on the choice of $\Db$. 
The force of our result is that there is {\em no} way of choosing a proper endomorphism of each spacetime's algebra in a natural way.

Our results on proper endomorphisms can be applied to discuss 
some simple comparisons between different theories.  
If $\Af\nto\Bf$ and $\Bf\nto\Af$ are subtheory embeddings, and at least one of these theories obeys our conditions, then 
one may easily show that both subtheory embeddings are in fact isomorphisms
(this is a strong analogue of the Cantor--Schr\"oder--Bernstein theorem of set theory). 
In particular, if there is any subtheory
embedding $\Af\nto\Bf$ that is not an isomorphism, then there can be {\em no} embedding 
$\Bf\nto\Af$.  For example, if $\Af$ obeys the conditions then there is no subtheory
embedding $\Af\otimes\Bf\nto\Af$ unless $\Bf$ is trivial.
There is another parallel with set theory: recall that
a set is finite precisely when there there is no injection from it into itself that is not a bijection; 
a notion of finiteness ascribed to Dedekind. In a general category, an object 
is said to be {\em Dedekind finite} if it has no monic proper endomorphisms (for example, finite-dimensional vector spaces). 
Our result may be paraphrased as indicating that energy compactness, together with our other assumptions, 
implies Dedekind finiteness in this sense for locally covariant QFTs. 

In Sec.~\ref{sect:examples} we show how the gauge group may be computed for
a system of finitely many free scalar fields with any given mass spectrum, both
as a classical theory and as a quantum field theory. As one
might expect, gauge transformations preserve sectors with different mass. 
Within each given mass sector, the gauge transformations act by orthogonal
transformations among the different fields with the same mass; in the massless
quantized case this is augmented by the freedom to add multiples of the unit element
of the algebra, resulting in a noncompact gauge group. This corresponds to
the broken symmetry of the Lagrangian under addition of a constant to the field. 
In both the classical and quantized cases, there are no proper endomorphisms;
in the quantized case there is a side condition that we restrict to endomorphisms
preserving the class of states with distributional $n$-point functions. 
The algebra of observables is also computed in the quantum case,
and we describe briefly how the theory of Sec.~\ref{sect:compactness}
can be applied to the Weyl algebra quantisation of the theory, in the
case where all fields are massive.

Open questions and directions resulting from these results include
the `top-down' analysis of superselection sectors mentioned above, and 
the possible extension of results on classification of subsystems~\cite{CarCon:2001b,CarCon:2005}
to curved spacetime. A key question is to understand the conditions under which there are
no accidental symmetries in Minkowski space, i.e., that the maximal DHR gauge group coincides
with the functorial definition. Finally we mention, as a related work, 
a paper of Ciolli, Ruzzi and Vasselli~\cite{CioRuzVas:2011} that
constructs a rather general covariant theory with given symmetry group, from which it is hoped
(and proved in some cases) that theories of interest can be obtained as specific representations
and which may even provide hints towards the inclusion of local gauge transformations.

\section{General framework}\label{sect:framework}
\subsection{Locally covariant theories}

We begin with a brief summary of the main ideas in the BFV framework,
as refined in~\cite{FewVer:dynloc_theory}.

\paragraph{Spacetimes} 
The category $\Loc$ has, as objects, quadruples of the form $\Mb= (\Mc,\gb,\ogth,\tgth)$,
where $(\Mc,\gb)$ is a nonempty smooth paracompact globally hyperbolic Lorentzian spacetime of dimension $n$ with at most finitely many connected components and $\ogth$ and $\tgth$ represent choices of orientation and
time-orientation. The spacetime dimension $n\ge 2$ is  fixed and the signature convention is
$+-\cdots -$. Morphisms in $\Loc$ are smooth isometric embeddings, preserving orientation and time-orientation, with causally convex image (i.e., containing all causal curves whose endpoints it contains); our notation will not distinguish the morphism from its underlying map
of manifolds.\footnote{Note, however, that the same map of manifolds can 
induce morphisms between many different pairs of objects in $\Loc$,
and that these morphisms are to be distinguished. A similar comment
applies to various other categories that will be discussed in this paper.} 
Note that causal convexity requires, in particular, that disjoint components of the image
are causally disjoint. The connected spacetimes form a full subcategory of $\Loc$, denoted $\Loco$. 

\paragraph{Theories as functors} Locally covariant theories can be described as covariant functors from $\Loc$ (or $\Loco$) to a
suitable category $\Phys$ of physical systems. Thus, to each $\Mb\in\Loc$ there is
an object $\Af(\Mb)$ of $\Phys$, and to each morphism $\psi:\Mb\to\Nb$ of $\Loc$
there is a morphism $\Af(\psi)$ of $\Phys$; we require $\Af(\psi\circ\varphi)=  \Af(\psi)\circ\Af(\varphi)$
for all composable $\psi$, $\varphi$, and $\Af(\id_\Mb) =\id_{\Af(\Mb)}$ for each $\Mb$. 
BFV mainly studied the case where $\Phys$ is the category $\Alg$ of unital $*$-algebras,\footnote{We exclude the zero algebra, a convention that was also used, albeit
unstated, in~\cite{FewVer:dynloc_theory, FewVer:dynloc2}.} with unit-preserving $*$-monomorphisms as the morphisms, or its full subcategory $\CAlg$ of $C^*$-algebras. When discussing classical fields, 
we will employ categories of (pre)symplectic spaces. {\em It will always be assumed that all 
morphisms of $\Phys$ are monic.}

The BFV framework provides a natural description of local physics. 
Let $\OO(\Mb)$ be the set of all open causally convex subsets $O$ of $\Mb$ with finitely many connected components. If $O\in\OO(\Mb)$ is 
nonempty, we define $\Mb|_O$ to be
the set $O$ equipped with the metric and (time)-orientation induced from $\Mb$ and regarded as a spacetime in its own right, with $\iota_{\Mb;O}:\Mb|_O\to\Mb$ being the canonical inclusion morphism. Then we may define $\Af^\kin(\Mb;O)$
to be the image of the map $\Af(\iota_{\Mb;O})$.\footnote{Not
all categories associate `images' to morphisms, although all those we use do. In more general situations, the local physics is better described as
the subobject of $\Af(\Mb)$ determined by $\Af(\iota_{\Mb;O})$.} (There are other ways of defining local physics, for example, the dynamics-based approach
introduced in~\cite{FewVer:dynloc_theory} -- note that in that reference
$\Af^\kin(\Mb;O)$ was taken to be the domain of
$\Af(\iota_{\Mb;O})$, rather than its image.)

\paragraph{Relative Cauchy evolution} Dynamics is incorporated in a
very natural way. A morphism $\psi:\Mb\to\Nb$ whose image contains
a Cauchy surface for $\Nb$ is described as {\em Cauchy}; a theory $\Tf:\Loc\to\Phys$
satisfies the {\em timeslice axiom} if $\Tf$ maps Cauchy morphisms to isomorphisms in $\Phys$.
If $\Mb=(\Mc,\gb,\ogth,\tgth)$, then to every smooth metric perturbation $\hb$ of compact support
for which $\Mb[\hb]=  (\Mc,\gb+\hb,\ogth,\tgth_\hb)$ is also globally hyperbolic (where
$\tgth_\hb$ is determined by requiring agreement with $\tgth$ outside   $\supp \hb$) 
there is an automorphism $\rce_\Mb[\hb]$ of $\Tf(\Mb)$, called the {\em relative Cauchy evolution}, that
compares the dynamics in $\Mb$ with that in $\Mb[\hb]$.
The details of the construction can be found in BFV and (slightly reformulated)~\cite{FewVer:dynloc_theory} and will not be repeated here. In circumstances
where the relative Cauchy evolution can be functionally differentiated with respect to $\hb$, the functional derivative can be interpreted in terms of a stress-energy tensor 
(see, e.g., Sec.~\ref{sec:states}), 
an interpretation supported by computations in specific models. 

\paragraph{Examples} Many standard models of quantum field theory
in curved spacetime have been formulated in the locally covariant framework, 
including the free scalar~\cite{BrFrVe03} and the Dirac quantum field~\cite{Sanders_dirac:2010}, and, importantly, the 
respective extended algebras of Wick products: see \cite{Ho&Wa01} for scalar fields (refined in \cite[\S 5.5.3]{BruFre_LNP:2009}) and~\cite{DHP_dirac:2009} for the Dirac case. Theories with gauge invariance 
also fit into the framework modulo important caveats relating to 
global issues, and at the time of writing, a definitive understanding is
yet to be reached. Relevant references include~\cite{DappLang:2012,FewsterHunt:2013,SandDappHack:2012,
BeniniDappiaggiSchenkel:2013}. A common theme is that injectivity of
the morphisms may be lost for certain observables of global nature, or
alternatively, that certain morphisms in $\Loc$ should be excluded 
from consideration. In this paper we work with injective morphisms on the basis that too much is lost if injectivity is dropped wholesale, and a clean characterisation of the `global' observables would be required to incorporate such ideas at the axiomatic level. Moreover, it is argued
that these pathologies might be removed in a fully interacting theory~\cite[Remark~4.15]{SandDappHack:2012}.

The quantum field theories mentioned can all be given
state spaces (see below), for example, based on Hadamard states. 
Treatments of classical theories include linear models described in 
symplectic spaces~\cite{FewVer:dynloc2} and, for general field theories~\cite{FreRej_BVclass:2012, BruFreLRib:2012arXiv1209.2148B}.

\subsection{Endomorphisms and automorphisms}\label{sec:endos_autos}

\paragraph{Definition and basic properties}
The functors from $\Loc$ to $\Phys$ form the objects of a category of locally covariant
theories, $\LCT$, introduced in~\cite{FewVer:dynloc_theory}. We will
use the notation $\LCT_\Phys$ if the target category is
not clear from context. The morphisms in this category
are natural transformations $\zeta:\Af\nto\Bf$: that is, to each $\Mb$ there is a morphism
$\zeta_\Mb:\Af(\Mb)\to\Bf(\Mb)$, so that  the equality $\zeta_\Nb\circ\Af(\psi) =\Bf(\psi)\circ\zeta_\Mb$ 
holds for every morphism $\psi:\Mb\to\Nb$.  The interpretation is that $\zeta$ embeds $\Af$ as a subtheory of $\Bf$, so an endomorphism
$\zeta:\Af\nto\Af$ of $\Af$ is an embedding
of $\Af$ as a subtheory of itself; specialising further, $\zeta$ is an automorphism of $\Af$ if each component $\zeta_\Mb$ is an isomorphism $\zeta_\Mb:\Af(\Mb)\to\Af(\Mb)$. The following observations are crucial.
\begin{Prop} \label{prop:basic}
Suppose $\Af:\Loc\to\Phys$, $\eta\in\End(\Af)$ and let $\Mb$
be any spacetime. Then, (a) we have
\begin{equation}\label{eq:basic}
\eta_\Mb\circ\Af(\iota_{\Mb;O}) = \Af(\iota_{\Mb;O})\circ\eta_{\Mb|_O}
\qquad\text{and}\qquad
\eta_\Mb\circ\Af(\psi) = \Af(\psi)\circ\eta_{\Mb}
\end{equation}
for every nonempty $O\in\OO(\Mb)$ and $\psi\in\End(\Mb)$;
(b) if $\Af$ also satisfies the timeslice axiom, then
\begin{equation}\label{eq:intertwine}
\rce_\Mb[\hb]\circ\eta_\Mb = \eta_\Mb\circ \rce_\Mb[\hb]
\end{equation}
in every spacetime $\Mb$ and for all permitted metric perturbations $\hb\in H(\Mb)$.
\end{Prop}
\begin{proof} 
Eq.~\eqref{eq:basic} is simply two instances of the
definition of naturality, while Eq.~\eqref{eq:intertwine} is a special case of~\cite[Proposition~3.8]{FewVer:dynloc_theory}. 
\end{proof}

Although part~(a) of the result is completely elementary, it immediately
tells us that endomorphisms act locally, and automorphisms act
strictly locally: if $\Phys=\Alg$ or $\CAlg$, for instance, $\Af^\kin(\Mb;O)$ 
is the image of $\Af(\iota_{\Mb;O})$ and we have
\begin{equation}\label{eq:local_action}
\eta_\Mb(\Af^\kin(\Mb;O))\subset \Af^\kin(\Mb;O),
\end{equation}
with equality if $\eta$ is an automorphism. Moreover, the second
part of~(a) asserts that endomorphisms of the theory commute with spacetime symmetries (indeed, even with spacetime endomorphisms). 
Thus two of the defining properties of an internal symmetry in AQFT are met, and generalised, by automorphisms of a locally covariant theory. In Sec.~\ref{sec:gauge_gp} we will see
how other standard properties are realised in representations. 

The fact that endomorphisms commute with relative Cauchy evolution
will be important when we come to classify them in particular models.
In circumstances where the relative Cauchy evolution may be differentiated with respect to the metric perturbation, Eq.~\eqref{eq:intertwine} asserts that endomorphisms preserve the
stress-energy tensor. 

On operational grounds, it is important to understand the
extent to which an endomorphism of a locally covariant theory
is determined by its behaviour in any single spacetime; 
put another way, if two endomorphisms have the same action 
in one spacetime, what can be said about their action in others? 
A full treatment requires additional assumptions (see below) but
we may make some preliminary observations:
\begin{Lem} \label{lem:tools1}
Consider a theory $\Af:\Loc\to\Phys$ obeying the timeslice condition. Let $\eta,\etap\in\End(\Af)$ and suppose that $\eta_\Mb=\etap_\Mb$ for
some spacetime $\Mb$. Then the following are true:
(i) if $\Lb\stackrel{\psi}{\to}\Mb$ then $\eta_\Lb=\etap_\Lb$; (ii)
if $\Mb\stackrel{\varphi}{\to}\Nb$ is Cauchy then
$\eta_\Nb=\etap_\Nb$; (iii) $\eta_\Lb=\eta'_{\Lb}$ 
for any spacetime $\Lb$ whose Cauchy surfaces are oriented-diffeomorphic to those of $\Mb|_O$ for
some $O\in\OO(\Mb)$. 
\end{Lem}
\begin{proof} (i) Because $\eta$ and $\etap$ are natural, we
have $\Af(\psi)\circ \eta_\Lb = \eta_\Mb\circ \Af(\psi) = \etap_\Mb\circ
\Af(\psi) = \Af(\psi)\circ \etap_\Lb$
and since $\Af(\psi)$ is monic, $\eta_\Lb=\etap_\Lb$. (The time slice property is
not required for this argument.) 
(ii) As $\Af(\varphi)$ is an isomorphism, we have $
\etap_\Nb = \Af(\varphi)\circ\etap_\Mb\circ\Af(\varphi)^{-1} = 
\Af(\varphi)\circ\eta_\Mb\circ\Af(\varphi)^{-1} = \eta_\Nb$
as required. For (iii), we use ``Cauchy wedge connectedness''~\cite[Proposition~2.4]{FewVer:dynloc_theory} (a 
formalisation of spacetime deformation arguments going back to~\cite{FullingNarcowichWald}) to obtain a chain of morphisms
\[
\Lb \xlongleftarrow{c} \Lb'
\xlongrightarrow{c} \Lb'' \xlongleftarrow{c} \Lb''' \xlongrightarrow{c} \Mb|_O \xlongrightarrow{\iota_{\Mb;O}} \Mb,
\]
in which a `c' above a morphism indicates that it is Cauchy. 
Starting at the right-hand end of this chain, where  $\eta_\Mb=\etap_\Mb$, we use parts~(i) and~(ii) to
move leftwards, deducing that the components of $\eta$ and $\etap$ agree in $\Mb_O$, $\Lb'''$ (using part~(i) twice), $\Lb''$ (using part~(ii)), 
$\Lb'$ (part~(i)) and finally $\Lb$ (part~(ii) again). 
\end{proof}

\paragraph{Additivity and the determination of an endomorphism from a single spacetime} 
The theories we will study satisfy additivity
properties of the type expected of field theories. Namely, 
in each spacetime $\Mb$, the object $\Af(\Mb)$ is generated in a suitable sense by its subobjects $\Af^\kin(\Mb;O)$ as $O$ runs over
a set of subspacetimes of $\Mb$. For the latter, we will use the 
{\em truncated multi-diamonds}~\cite[Definition~2.5]{FewVer:dynloc_theory}, which are sets of the form $\Nc\cap D_\Mb(B)$,
where $\Nc$ is an open globally hyperbolic neighbourhood of a Cauchy surface $\Sigma$ for $\Mb$, 
and $B$ is a union of finitely many disjoint subsets of $\Sigma$ each of which is a nonempty open ball in suitable local coordinates. Sets of the above form with
$\Nc=\Mb$ are called {\em multi-diamonds}. 

The sense in which the $\Af^\kin(\Mb;O)$ generate $\Af(\Mb)$ depends
on the category, and can be expressed abstractly using the notion of a
categorical {\em union}. A category $\Ct$ is said to have unions \cite[\S1.9]{DikranjanTholen} if, given
any family $(m_i)_{i\in I}$ of monic $\Ct$-morphisms $m_i:M_i\to A$, 
representing $\Ct$-subobjects of $A$, 
there exists a monic $m:M\to A$ such that (1) each $m_i$ factorises as $m_i=m\circ \tilde{m}_i$, and (2) given any $f:A\to B$ and a monic $n:N\to B$ such that every $f\circ m_i$ factorises as
$n\circ\tilde{n}_i$, then there is a unique morphism $\tilde{f}:M\to N$
such that $n\circ \tilde{f}=f\circ m$ and $\tilde{f}\circ \tilde{m}_i=\tilde{n}_i$ for all $i\in I$. In other words, 
commutativity of the outer portion of the diagram 
\begin{equation}\label{eq:property(2)}
\begin{tikzpicture}[baseline=0.5em, description/.style={fill=white,inner sep=2pt}]
\matrix (m) [ampersand replacement=\&,matrix of math nodes, row sep=3em,
column sep=3em, text height=1ex, text depth=0.25ex]
{M_i \&  M \& A   \\
       \& N \&  B    
\\};
\path[->]
(m-1-1)  edge node[above] {$ \tilde{m}_i  $} (m-1-2)
             edge node[below,sloped] {$ \tilde{n}_i $} (m-2-2)
(m-1-2)  edge[style=dashed] node[right] {$ \tilde{f} $} (m-2-2)
             edge node[above] {$ m $} (m-1-3)
(m-1-3)  edge node[right] {$ f $} (m-2-3)
(m-2-2)  edge node[above] {$ n $} (m-2-3);
\end{tikzpicture}
\end{equation}
for each $i\in I$ (with the understanding that $m_i=m\circ \tilde{m}_i$) entails the existence of a unique $\tilde{f}$ making the diagram  commute in full.
The union subobject $m:M\to A$ is defined up to
isomorphism and we write
\[
m\cong \bigvee_{i\in I} m_i : \bigvee_{i\in I} M_i\to A.
\]
(See \cite{DikranjanTholen} and~\cite[Appendix~B]{FewVer:dynloc_theory}
for more details.) Among the categories we employ for $\Phys$, 
both $\Alg$ and $\CAlg$ have unions, corresponding to the $(C)^*$-subalgebra generated by a family of $(C)^*$-subalgebras. The same is true of the category of complexified presymplectic spaces 
appearing in Sec.~\ref{sect:examples} (linear span of presymplectic subspaces). However,  the category of symplectic spaces does not
have unions---note that the linear span of symplectic subspaces need not be symplectic. We may now give a precise statement of additivity.
\begin{Def} 
A theory $\Af:\Loc\to\Phys$ is said to be
{\em additive} if $\Phys$ has unions and, for each spacetime $\Mb$, 
\[
\Af(\Mb)=\bigvee_{D\subset \Mb} \Af^\kin(\Mb;D),\qquad
\text{or, more precisely,}~\quad \id_{\Af(\Mb)} \cong 
\bigvee_{D\subset\Mb} \Af(\iota_{\Mb;D}),
\]
where the union runs over the set of all truncated multi-diamond subsets of $\Mb$.
\end{Def}
In particular, any {\em dynamically local} theory is additive in this sense~\cite[Theorem~6.3]{FewVer:dynloc_theory}.

It will be convenient to consider categories where the
existence of unions would either be tedious to demonstrate or even fails, but where there is a related category that does have unions. In such circumstances, the following generalised definition is useful. 
\begin{Def}
A theory $\Af:\Loc\to\Phys$ is said to be
{\em $\Uf$-additive} if $\Uf$ is a faithful functor\footnote{That is, $\Uf$ is injective as a
map of morphisms.} $\Uf:\Phys\to\Phys'$,
where $\Phys'$ is a category possessing unions and all of whose morphisms are monic, such that $\Uf\circ\Af$ is additive. 
\end{Def}
$\Uf$-additivity includes additivity as a special case, if $\Phys$ has unions, by taking $\Uf$ to be the identity functor on $\Phys$.

We will need a simple technical lemma, applying if $\Ct$ has both
unions and {\em equalizers} for arbitrary pairs of morphisms. 
Here, an equalizer of $f,g:B\to C$ in $\Ct$ is a morphism $h:A\to B$ such that $f\circ h=g\circ h$ and satisfying the property that, if $k$ is any morphism with $f\circ k=g\circ k$ then $k$ factorizes uniquely via $h$, i.e., $k=h\circ m$ for a unique morphism $m$. Equalizers are determined up to isomorphism by this definition; we write $h\cong \eq(f,g)$. The categories $\Alg$ and $\CAlg$ have
equalizers: morphisms $\alpha,\beta:\Ac\to\Bc$ are equalized by the inclusion map in $\Ac$ of the $(C)^*$-subalgebra of $\Ac$ on which $\alpha$ and $\beta$ agree. 
\begin{Lem}\label{lem:union_invariance}
Suppose $\Ct$ has unions and equalizers. Let
 $(m_i)_{i\in I}$ be a class-indexed family of subobjects of $A\in \Ct$ with union $m:M\to A$. If morphisms $g$ and $h$ obey $g\circ m_i=h\circ m_i$ for all $i\in I$, then $g\circ m=h \circ m$. If, additionally, 
$m$ is an isomorphism, then $g=h$. 
\end{Lem}
\begin{proof} We have $g\circ m_i=h\circ m_i$ and hence a factorisation $m_i=\eq(g,h)\circ \tilde{n}_i$ for each $i\in I$. Setting $B=A$, $f=\id_A$ and $N$ equal to the domain of $\eq(g,h)$, the outer portion of diagram~\eqref{eq:property(2)} commutes for all $i\in I$, and there is therefore a morphism $\tilde{f}$ to make the diagram commute in full. In particular,
$\eq(g,h)\circ\tilde{f} = m$, so
$g\circ m = g\circ \eq(g,h)\circ\tilde{f}= h\circ \eq(g,h)\circ\tilde{f}= m$
as required. The last statement is immediate (it would be enough
for $m$ to be epic). 
\end{proof}

We can now complete the discussion begun in Lemma~\ref{lem:tools1}.
\begin{Thm} \label{thm:reduce_to_Mink}
Suppose $\Af$ obeys the timeslice axiom and is $\Uf$-additive with respect to 
$\Uf:\Phys\to\Phys'$, where $\Phys'$ has unions and equalizers. Then every $\eta\in\End(\Af)$ is uniquely determined by its component $\eta_{\Mb}$ in any given spacetime $\Mb$.
\end{Thm}
{\noindent\bf Remark.} 
In particular, the conclusion holds if $\Phys$ has
unions and equalizers and $\Af$ is additive, by taking $\Uf$ to be the identity functor.

\begin{proof} Suppose $\etap\in\End(\Af)$ agrees with $\eta$ in $\Mb$, i.e., $\etap_{\Mb}=\eta_{\Mb}$.   If $\Nb$ is
any spacetime and $D$ is any truncated multi-diamond in $\Nb$ then
$\Nb|_D$ has Cauchy surfaces oriented-diffeomorphic to any truncated
multi-diamond in $\Mb$ with the same number of connected components as $D$. 
Accordingly, $\eta_{\Nb_D} =\etap_{\Nb_D}$ by Lemma~\ref{lem:tools1}(iii), and the naturality of $\eta$ and $\etap$ gives
\[
\eta_\Nb\circ\Af(\iota_{\Nb;D}) =
\Af(\iota_{\Nb;D})\circ\eta_{\Nb|_D} = \Af(\iota_{\Nb;D})\circ\etap_{\Nb|_D} = \etap_\Nb\circ\Af(\iota_{\Nb;D}).
\]
Applying $\Uf$, we have $\Uf(\eta_\Nb)\circ\Uf(\Af(\iota_{\Nb;D})) =\Uf(\etap_\Nb)\circ \Uf(\Af(\iota_{\Nb;D}))$. By Lemma~\ref{lem:union_invariance}, it follows that $\eta_{\Nb}=\etap_{\Nb}$  because  $\bigvee_{D\subset \Nb}\Uf(\Af(\iota_{\Nb;D}))$
is an isomorphism and $\Uf$ is faithful. As $\Nb$ was arbitrary, $\eta=\etap$. 
\end{proof}

\subsection{States and twisted locality}\label{sec:states}

The discussion of the previous subsections was conducted
quite abstractly, in order to emphasise the general applicability 
of the ideas. In order to make contact with quantum field theory,
we now describe more specific categories of physical systems
that incorporate not only $*$-algebras, but also states,
and allow for the Bose/Fermi distinction. Our discussion of state 
spaces is based almost entirely on that in BFV, but the discussion
of twisted locality is new and, in fact, is made possible by the discussion above.

\paragraph{States} By a {\em state space} for an algebra $\Ac\in\Alg$, we mean a subset $\Sc$ of normalized
positive linear functionals on $\Ac$ that is closed under convex linear combinations, and 
operations induced by $\Ac$ [i.e.,
to each $\omega\in\Sc$ and $B\in\Ac$ with $\omega(B^*B)> 0$, the state $\omega_B(A):=\omega(B^*AB)/\omega(B^*B)$ is also an element of $\Sc$]. BFV raised this idea to the functorial level: along with
a functor $\Af:\Loc\to\Alg$, they considered a
contravariant functor $\Sf$ from $\Loc$ to a suitable category of state spaces, with the property that each $\Sf(\Mb)$ is a state space for $\Af(\Mb)$ and that each $\Sf(\psi)$ is an appropriate restriction of the dual map $\Af(\psi)^*$. Then $\Sf$ is called a state space for $\Af$.
The state space may be given various additional attributes~\cite{BrFrVe03}; 
in particular, we say that $\Sf$ is {\em faithful} if
\[
\bigcap_{\omega\in\Sf(\Mb)}\ker\pi_\omega = \{0\},
\]
where $\pi_\omega$ is the GNS representation of $\Af(\Mb)$ induced by $\omega$. Given
the other properties of a state space, faithfulness also implies\footnote{If $\omega(A)=0$ for all $\omega\in\Sf(\Mb)$ then also $\omega(B^* AB)=0$ for
all $\omega\in\Sf(\Mb)$ and $B\in\Af(\Mb)$; polarising, $\omega(B^* AC )=0$ 
for all $\omega\in\Sf(\Mb)$ and $B,C\in\Af(\Mb)$, so $\pi_\omega(A)=0$ for every $\omega\in\Sf(\Mb)$.} that 
\[
\bigcap_{\omega\in\Sf(\Mb)}\ker \omega = \{0\}.
\]
In the $C^*$-case, the state space is said to be {\em locally quasi-equivalent} if, 
for every spacetime $\Mb$, relatively compact $O\subset\Mb$ and states $\omega_i\in\Sf(\Mb)$ ($i=1,2$), the
GNS representations $\langle \HH_{\omega_i},\pi_{\omega_i},\Omega_i\rangle$ restrict to quasi-equivalent 
representations of $\Af^\kin(\Mb;O)$, i.e., 
the sets of states on $\Af^\kin(\Mb;O)$ induced by
density matrices on $\HH_1$ and $\HH_2$ coincide. 

The following simple observation will be useful. 
\begin{Lem} \label{lem:invariant_states}
Suppose $\eta\in\End(\Af)$, where $\Af:\Loc\to\Alg$ or $\CAlg$. If $\psi\in\Aut(\Mb)$ and $\omega$ is an $\Af(\psi)$-invariant
state on $\Af(\Mb)$ then $\eta_\Mb^*\omega$ is also invariant.
\end{Lem}
\begin{proof} 
$\Af(\psi)^*(\eta_\Mb^*\omega)= \omega\circ\eta_\Mb\circ\Af(\psi) =
\omega\circ\Af(\psi)\circ\eta_\Mb = \omega\circ\eta_\Mb =
\eta_\Mb^*\omega$.  
\end{proof}

\paragraph{Graded algebras, states and (twisted) locality}
We combine the description of the algebras and their state spaces
as a single mathematical object. At the same time we build in the possibility of describing Bose and Fermi
statistics by considering a category
of graded algebras with states, $\grAS$, whose objects are triples $\langle \Ac,\gamma,\Sc\rangle$ consisting of an algebra
$\Ac$ together with a choice of state space $\Sc$ for $\Ac$, and an
involutive automorphism $\gamma$ of $\Ac$ obeying $\gamma^*\Sc=\Sc$, which determines a $\ZZ_2$-grading. A morphism between
triples $\langle \Ac,\gamma,\Sc\rangle$ and $\langle \Bc,\delta,\Tc\rangle$ is
determined by any $\alpha:\Ac\to\Bc$ in $\Alg$ with the property that $\alpha\circ \gamma=
\delta\circ\alpha$ and $\alpha^*\Tc\subset \Sc$. The association between
any $\grAS$ morphism and its underlying $\Alg$ morphism determines
a faithful functor $\Uf:\grAS\to\Alg$ such that $\Uf(\langle \Ac,\gamma,\Sc\rangle)=\Ac$. We also have an obvious analogue $\grCAS$, 
obtained by replacing $\Alg$ by $\CAlg$ throughout and to
which the following remarks apply {\em mutatis mutandis}.

A theory $\Xf\in \LCT_\grAS$ assigns 
a triple $\Xf(\Mb) = \langle \Af(\Mb), \gamma_\Mb, \Sf(\Mb)\rangle\in\grAS$ to each $\Mb\in\Loc$,
and to each morphism $\psi:\Mb\to\Nb$ a corresponding morphism
in $\grAS$. It follows immediately that $\Uf\circ \Xf$ is a theory
in $\LCT_\Alg$, with $\Uf\circ\Xf(\Mb) = \Af(\Mb)$; similarly, 
the $\Sf(\Mb)$ form a state space for $\Af$. Moreover, the
$\gamma_\Mb$ form the components of an automorphism
$\gamma\in\Aut(\Af)$ obeying $\gamma^2=\id_\Af$ and under
which $\Sf$ is invariant. 

A subtheory embedding between $\Xf = \langle \Af,\gamma,\Sf\rangle$
and $\Yf = \langle \Bf,\delta,\Tf\rangle$ in $\LCT_\grAS$ is, as usual, 
a natural transformation $\zeta:\Xf\nto\Yf$. The morphisms $\Uf(\zeta_\Mb)$
form the components of a natural $\Uf(\zeta):\Af\nto\Bf$, such that
$\Uf(\zeta)\circ\gamma=\delta\circ\Uf(\zeta)$ and $\Uf(\zeta)^*\Tf$ is
a subfunctor of $\Sf$.\footnote{That is, $\Uf(\zeta_\Mb)^*(\Tf(\Mb))\subset \Sf(\Mb)$, 
and $\Uf(\zeta_\Mb)^*\Tf(\psi)$ is a restriction of $\Sf(\psi)$ for all $\Mb$, $\psi$.} 
As $\Uf$ is faithful, $\zeta\mapsto\Uf(\zeta)$ determines an isomorphism
\begin{equation}\label{eq:Aut_iso}
\Aut(\langle \Af,\gamma,\Sf\rangle) \cong \{
\eta\in\Aut(\Af): \eta\circ\gamma =\gamma\circ\eta,~\eta^*\Sf=\Sf\}
\end{equation}
so the introduction of the grading and state space can break the symmetry
group of $\Af$ to a subgroup of the centralizer of $\gamma$ in $\Aut(\Af)$. 
As $\gamma$ is an element of the right-hand side of~\eqref{eq:Aut_iso}, 
it follows that there is a (unique) $\hat{\gamma}\in\Aut(\langle \Af,\gamma,\Sf\rangle)$ such that $\Uf(\hat{\gamma}) =\gamma$. Furthermore, 
$\hat{\gamma}^2 =\id_{\langle \Af,\gamma,\Sf\rangle}$, and $\hat{\gamma}$
is evidently central in $\Aut(\langle \Af,\gamma,\Sf\rangle)$. (In
passing, note that if we replace $\Sf$ by an extended
state space $\widetilde{\Sf}$ with
\[
\widetilde{\Sf}(\Mb) = \co\bigcup_{\eta\in\Aut(\Af)} \eta_\Mb^*\Sf(\Mb),
\]
where $\co$ denotes closure under (finite) convex linear combinations,
the gauge group will coincide with the centralizer of $\gamma$. If
$\gamma$ is central in $\Aut(\Af)$, this would also ensure that
$\Aut(\langle \Af,\gamma,\widetilde{\Sf}\rangle)\cong \Aut(\Af)$.) 

The automorphism $\gamma$ may be used to define a graded commutator on $\Af(\Mb)$ by
\[
\bbLbrack A, B \bbRbrack= A B -(-1)^{\sigma\sigma'} BA
\]
for $A,B\in\Af(\Mb)$ such that $\gamma_\Mb(A)=(-1)^\sigma A$, $\gamma_\Mb(B)=(-1)^{\sigma'} B$ ($\sigma,\sigma'\in\{0,1\}$), and extended by linearity. As $\hat{\gamma}$ is central in $\Aut(\langle \Af,\gamma,\Sf\rangle)$, 
it follows that the graded commutator is equivariant in the sense that
\[
\bbLbrack \Uf(\zeta)_\Mb A,\Uf(\zeta)_\Mb B \bbRbrack= 
\Uf(\zeta)_\Mb \bbLbrack  A, B \bbRbrack
\]
for all $\zeta\in\Aut(\langle \Af,\gamma,\Sf\rangle)$, $\Mb\in\Loc$ and
$A,B\in\Af(\Mb)$. 
The theory $\langle \Af,\gamma,\Sf\rangle$ can then be said to obey {\em twisted locality} if 
\[
\bbLbrack\Af^\kin(\Mb;O_1),\Af^\kin(\Mb;O_2)\bbRbrack= \{0\}
\]
whenever $O_i\in\OO(\Mb)$ are causally disjoint, which implements
standard commutation relations for a mixture of bosonic and fermionic
degrees of freedom and reduces to commutation
at spacelike separation if $\gamma=\id_\Af$.

We briefly connect these new structures with some of the
ideas in the previous subsections. First, if $\langle\Af,\gamma,\Sf\rangle$ obeys the timeslice axiom and the state space
$\Sf$ is faithful for $\Af$, the connection between the
relative Cauchy evolution and the stress-energy tensor can
be made more specific. It is easily seen that $\Uf$ maps the
relative Cauchy evolution of $\langle\Af,\gamma,\Sf\rangle$ to that
of $\Af$: $\Uf(\rce^{(\langle\Af,\gamma,\Sf\rangle)}_\Mb[\hb]) = \rce^{(\Af)}_\Mb[\hb]$; moreover, 
$\rce_\Mb[\hb]\circ\gamma_\Mb=\gamma_\Mb\circ
\rce_\Mb[\hb]$ and $\rce_\Mb[\hb]^*\Sf(\Mb)=\Sf(\Mb)$. The relative Cauchy evolution is said to be weakly differentiable with respect to $\Sf(\Mb)$ on all $A\in\Af(\Mb)$, if for each smooth $1$-parameter family $\lambda\mapsto \hb(\lambda)\in H(\Mb)$, there exists a (unique, due to faithfulness) element, denoted $[\Tb_\Mb(\fb),A]\in
\Af(\Mb)$ such that
\[
\omega([\Tb_\Mb(\fb),A]) = 2i\left.\frac{d}{d\lambda} 
\omega(\rce_\Mb[\hb(\lambda)]A)\right|_{\lambda=0}
\]
for all $\omega\in\Sf(\Mb)$, where $\fb=\dot{\hb}(0)$. This defines
a stress-energy tensor as a (possibly outer) symmetric derivation on $\Af(\Mb)$.
Under these circumstances, suppose that $\eta\in\End(\langle\Af,\gamma,\Sf\rangle)$. Then we may  differentiate the identity
$\omega(\rce_\Mb[\hb(\lambda)]\circ \eta_\Mb A) = \omega(\eta_\Mb\circ \rce_\Mb[\hb(\lambda)] A) =
(\eta_\Mb^*\omega)(\rce_\Mb[\hb(\lambda)] A)$ (using Proposition~\ref{prop:basic}) and again use faithfulness of $\Sf$ to obtain 
\[
[\Tb_\Mb(\fb),\eta_\Mb A] = \eta_\Mb [\Tb_\Mb(\fb),A] \qquad\forall A\in\Af(\Mb),
\]
so the stress-energy derivation commutes with all endomorphisms,
and in particular with the grading $\gamma_\Mb$. This makes precise
the sense in which endomorphisms preserve the stress-energy tensor,
and shows that the latter is necessarily Bosonic (as one would expect).

Finally, a theory $\langle \Af,\gamma,\Sf\rangle\in\LCT_\grAS$ is $\Uf$-additive if and only if $\Af$ is additive in $\LCT_\Alg$.
Subject to that condition and the timeslice property, Theorem~\ref{thm:reduce_to_Mink} applies to $\langle \Af,\gamma,\Sf\rangle$ and permits the determination of its endomorphisms
through the action in any individual spacetime.

\section{The gauge group and algebra of observables}

\subsection{Gauge group}\label{sec:gauge_gp}

In this section we study the automorphism group of theories
in $\LCT_\grAS$ and $\LCT_\grCAS$ in terms of the GNS representations
of the underlying algebras induced by their state spaces. This
makes direct contact with, and again generalises, the global gauge group of Minkowski space AQFT~\cite{DHRi}. 

Throughout this subsection we consider a fixed
theory $\langle \Ff,\gamma,\Sf\rangle$ which obeys the timeslice
axiom and $\Uf$-additivity as described above, and write
$G = \Aut(\langle\Ff,\gamma,\Sf\rangle)$. We now endow $G$ with a
topology and investigate its properties. 
\begin{Def} The {\em natural weak topology} on $G$ is
 the weakest group topology in which 
$G\owns\eta\mapsto \omega(\eta_\Mb F)$ is continuous for all $\Mb\in\Loc$,
$\omega\in\Sf(\Mb)$ and $F\in\Ff(\Mb)$.
\end{Def}

\begin{Prop} 
If $\Sf$ is faithful then the natural weak topology of $G$ is Hausdorff (and  therefore finer than the indiscrete topology provided $G$ is nontrivial). 
\end{Prop}
\begin{proof} Suppose without loss that $G$ is nontrivial and let $\eta,\zeta\in G$ be arbitrary, with $\eta\neq\zeta$. Theorem~\ref{thm:reduce_to_Mink} entails that $\eta_\Mb\neq\zeta_\Mb$, so there exists $F\in \Ff(\Mb)$ such that $\eta_{\Mb}F\neq 
\zeta_{\Mb}F$. As $\Sf$ is faithful, there is $\omega\in\Sf(\Mb)$ such that $\omega(\eta_{\Mb}F) \neq \omega(\zeta_{\Mb} F)$. Thus the topology separates $\eta$ and $\zeta$ and is therefore Hausdorff. \end{proof}

\begin{Prop}\label{prop:urep}
Suppose $\omega\in\Sf(\Mb)$ is gauge-invariant, i.e., $\eta_\Mb^*\omega=\omega$
for all $\eta\in G$, and induces a faithful GNS representation
$(\Hc_\omega,\DD_\omega,\pi_\omega,\Omega_\omega)$ of $\Ff(\Mb)$. Then
there is a faithful and strongly continuous representation $G\owns\eta\mapsto U_\eta$  such that
$\pi_\omega(\eta_\Mb F) = U_\eta \pi_\omega(F)U_\eta^{-1}$ and  $U_\eta\Omega_\omega =\Omega_\omega$ ($\eta\in G$, $F\in\Ff(\Mb)$)
and which acts strictly locally, that is
$U_\eta\Ff^\kin(\Mb;O)U_\eta^{-1} =\Ff^\kin(\Mb;O)$ for all
nonempty $O\in\OO(\Mb)$. Moreover, if $\pi_\omega$ is irreducible, the representation of $G$ commutes with the unitary representation of any spacetime automorphism under which $\omega$ is invariant. 
\end{Prop}
\begin{proof} The existence of the unitary representation is immediate from
gauge-invariance of $\omega$. If $U_\eta=\II_{\Hc_\omega}$, then $\eta_\Mb F = F$
for all $F\in\Ff(\Mb)$, because $\pi_\omega$ is faithful. Thus $\eta$ and 
$\id_{\langle\Ff,\gamma,\Sf\rangle}$
have equal components in $\Mb$ and are therefore equal by Theorem~\ref{thm:reduce_to_Mink},
so $\eta\mapsto U_\eta$ is faithful. By definition of the natural weak topology of $G$ and closure of the state space under operations, the maps 
$\eta\mapsto \omega((\II+\lambda A )^*(\eta_\Mb B)(\II +\lambda A ))$ are continuous for all $A,B\in\Ff(\Mb)$ and all $\lambda\in\CC$ of sufficiently small modulus. Expanding in $\lambda$ and $\bar{\lambda}$, we deduce that 
the maps $\eta\mapsto \omega(A^*\eta_\Mb B)$ and $\eta\mapsto \omega((\eta_\Mb B)A)$ [and $\eta\mapsto \omega(A^*(\eta_\Mb B) A)$]
are continuous for any $A,B\in\Ff(\Mb)$, and hence  that
$\eta\mapsto U_\eta$ is strongly continuous on the dense domain $\pi_\omega(\Ff(\Mb))\Omega_\omega$. An $\epsilon/3$ argument completes the proof of strong continuity. 
The strict locality of the representation follows immediately from
Proposition~\ref{prop:basic}(a), cf.\ Eq.~\eqref{eq:local_action}.

Now let $H_\omega$ be the (possibly trivial) subgroup of $\Aut(\Mb)$ leaving $\omega$ invariant, i.e., $\Af(\psi)^*\omega=\omega$ for all $\psi\in H_\omega$. Then there is also a unitary representation $H_\omega\owns
\psi\mapsto V_\psi$, with $\pi_\omega(\Ff(\psi) F) = V_\psi \pi_\omega(F)V_\psi^{-1}$ and 
$V_\psi\Omega_\omega =\Omega_\omega$ ($\psi\in H_\omega$, $F\in\Ff(\Mb)$). The computation
\[
U_\eta V_\psi\pi_\omega(F) V_\psi^* U_\eta^* = 
\pi_\omega(\eta_\Mb\circ\Ff(\psi)F) = \pi_\omega(\Ff(\psi)\circ\eta_\Mb
F) = V_\psi U_\eta\pi_\omega(F) U_\eta^* V_\psi^*
\]
shows that $U_\eta$ and $V_\psi$ commute up to phase, by irreducibility
of $\pi_\omega$; as both operators leave $\Omega_\omega$ invariant, 
it follow that $U_\eta$ and $V_\psi$ commute. 
\end{proof}

Under the hypotheses of Proposition~\ref{prop:urep}, we may define
a new group topology on $G$, namely the weakest in which the representation
$\eta\mapsto U_\eta$ is strongly continuous (this is necessarily
weaker than the original topology) -- we call this {\em the topology
induced by $\omega$}, or the {\em $\omega$-topology}.  Theorem~\ref{thm:tops} of Appendix~\ref{appx:topology} shows that the  natural weak topology is in fact equivalent to the topology induced by the Minkowski space vacuum state for theories in $\LCT_\grCAS$ obeying 
suitable conditions. Furthermore, Proposition~\ref{prop:Gmax} gives
conditions for these topologies to be compact. 

\subsection{Action on fields}

The gauge group of a theory acts in a natural way on the associated locally covariant
fields. Some relevant definitions are needed from BFV and~\cite{Fewster2007}, adapted slightly to our setting. To start, consider a general category of
physical systems $\Phys$, equipped with a functor $\Vf:\Phys\to\Set$, the category of sets and (not necessarily injective) functions. Fix a functor $\Df:\Loc\to\Set$.
Then any natural transformation $\Phi:\Df\nto\Vf\circ\Tf$ will be described as a field of type $\Df$ associated with $\Tf$. That is, to each $\Mb$ there is a function
$\Phi_\Mb:\Df(\Mb)\to\Vf(\Tf(\Mb))$, (not assumed to be injective) such that
\[
\Vf(\Tf(\psi))\Phi_\Mb(f) = \Phi_\Nb(\Df(\psi)f) 
\]
for each $\psi:\Mb\to\Nb$. We use $\Fld(\Df,\Tf)$ to denote the set of all such fields (suppressing $\Vf$ from the notation). For example, with $\Df(\Mb)=\CoinX{\Mb}$ (as a set) and $\Df(\Mb\stackrel{\psi}{\to}\Nb)$ given by the push-forward $\Df(\Mb\stackrel{\psi}{\to}\Nb)=\psi_*$, where
\begin{equation}\label{eq:psistar}
(\psi_* f )(p) = \begin{cases} f(\psi^{-1}(p)) & p\in \psi(\Mb)\\ 0 &\text{otherwise}
\end{cases}\qquad (f\in\CoinX{\Mb},~p\in\Nb),
\end{equation}
we obtain the scalar fields associated with the theory. Fields indexed by (possibly distributional) sections in more general
bundles can also be described in a similar way -- see, e.g.,~\cite{Fewster2007} -- we restrict to the scalar case in this 
subsection for simplicity.  The functor $\Vf$ is usually
obvious. For $\Phys=\Alg$ or $\CAlg$, we take the forgetful functor sending each algebra to its underlying set and each morphism to the underlying function; for $\grAS$ and $\grCAS$, we use the functor sending $\langle\Ac, \gamma,\Sc\rangle$ to the underlying set of $\Ac$ and morphisms to the underlying functions. 

For a theory $\langle\Ff,\gamma,\Sf\rangle\in\LCT_\grAS$,
the set $\Fld(\Df,\langle\Ff,\gamma,\Sf\rangle)$ (and, similarly,
$\Fld(\Df,\Ff)$) may be given
a unital $*$-algebra structure under pointwise operations inherited from the algebras $\Ff(\Mb)$~\cite{Fewster2007}. 
Thus $(\Phi+\lambda\Psi)_\Mb(f) = \Phi_\Mb(f) + \lambda\Psi_\Mb(f)$, $(\Phi\Psi)_\Mb(f) = 
\Phi_\Mb(f)\Psi_\Mb(f)$, $(\Phi^*)_\Mb(f)= \Phi_\Mb(f)^*$,\footnote{If $\Phi$ is a linear
field, this definition makes $\Phi^*$ conjugate linear.} and the unit field is
$\II_\Mb(f) = \II_{\Ff(\Mb)}$, for all $f\in\CoinX{\Mb}$, $\Mb\in\Loc$.
If $\langle\Ff,\gamma,\Sf\rangle\in\LCT_\grCAS$ 
then
\[
\|\Phi\| = \sup_{\Mb\in\Loc}\sup_{f\in\Df(\Mb)} \|\Phi_\Mb(f)\|_{\Ff(\Mb)}
\]
is a $C^*$-norm on the $*$-subalgebra $\Fld^\infty(\Df,\Ff)$ on which it is finite. (Some
set-theoretical niceties are glossed over here; see~\cite{Fewster2007}.)

These abstract algebras of fields carry an action of the automorphism group $G = \Aut(\langle\Ff,\gamma,\Sf\rangle)$ in an obvious way. Given any $\eta\in G$, and $\Phi\in\Fld(\Df,\langle\Ff,\gamma,\Sf\rangle)$, define $\eta\cdot\Phi$ by
\[
(\eta\cdot\Phi)_\Mb(f) = \eta_\Mb\Phi_\Mb(f)\qquad (f\in\CoinX{\Mb},~\Mb\in\Loc).
\]
This is easily seen to define a field $\eta\cdot\Phi\in\Fld(\Df,\langle\Ff,\gamma,\Sf\rangle)$ by the calculation
\begin{align*}
\Ff(\psi)(\eta\cdot\Phi)_\Mb(f) &= \Ff(\psi)\circ \eta_\Mb(\Phi_\Mb(f)) = 
\eta_\Nb \circ \Ff(\psi) \Phi_\Mb(f) = \eta_\Nb \Phi_\Nb(\psi_*f) \\
&= (\eta\cdot\Phi)_\Nb(\psi_*f) 
\end{align*}
for any $\psi:\Mb\to\Nb$, $f\in\CoinX{\Mb}$. Moreover, the action of $\eta$ is
evidently a $*$-automorphism of $\Fld(\Df,\langle\Ff,\gamma,\Sf\rangle)$ and gives a group homomorphism
$G\mapsto \Aut(\Fld(\Df,\langle\Ff,\gamma,\Sf\rangle))$ [and a $C^*$-automorphism of $\Fld^\infty(\Df,\langle\Ff,\gamma,\Sf\rangle)$,
and corresponding group homomorphism, if relevant]. 
Endowing $\Fld(\Df,\langle\Ff,\gamma,\Sf\rangle)$ with the weakest topology in which every function
$\Phi\mapsto \omega(\Phi_\Mb(f))$ ($\Mb\in\Loc$, $f\in\Df(\Mb)$, $\omega\in\Sf(\Mb)$) is continuous,
this action of $G$  is continuous with respect to the natural weak topology.\footnote{$\eta \mapsto\eta\cdot\Phi$ is continuous iff the functions $\eta\mapsto \omega((\eta\cdot\Phi)_\Mb(f))=
\omega(\eta_\Mb\Phi_\Mb(f))$ are continuous, which they are by definition.}

In particular, this gives a continuous linear representation of $G$ on $\Fld(\Df,\langle\Ff,\gamma,\Sf\rangle)$, regarded as a 
vector space. A {\em multiplet of fields} can now be defined as any subspace of $\Fld(\Df,\langle\Ff,\gamma,\Sf\rangle)$
transforming under an indecomposable representation of $G$, and 
every field can be associated with an equivalence class of $G$-representation.
Let $\rho, \sigma$ be the equivalence classes corresponding to fields
$\Phi$, $\Psi$. Then  $\Phi^*$ transforms in the
complex conjugate representation $\bar{\rho}$ to $\rho$, any linear combination of $\Phi$ and
$\Psi$ transforms in a subrepresentation of a quotient of $\rho\oplus\sigma$, 
and $\Phi\Psi$ and $\Psi\Phi$ transform in (possibly different)
subrepresentations of quotients of $\rho\otimes\sigma$. The quotients
reflect any algebraic relationships among the fields in the multiplets of
$\Phi$, $\Psi$ under the linear combination or product. For example, if $\Phi$ and $\Psi$
belong to a common multiplet, then their linear combinations belong to the same multiplet. 

The fact that both $\sigma$ and $\bar{\sigma}$ appear expresses the 
fundamental particle--antiparticle symmetry of quantum field theory.
Algebras of bi-local and multi-local fields can be defined, and similar 
comments apply to them.

\subsection{The algebra of observables}\label{sect:alg_of_obs}

In AQFT, the local observables are precisely those elements of the
local field algebras that are fixed points under the gauge group. 
An analogous construction may be carried out for any theory
$\Ff:\Loco\to\Phys$,\footnote{For the moment, we restrict to connected spacetimes;
see comments below.} provided that $\Phys$ has equalizers 
over arbitrary families of morphisms: 
in each $\Mb$, let $\alpha_\Mb$ be an equalizer for all the morphisms
$\eta_\Mb$ where $\eta\in\Aut(\Ff)$. Thus
$\eta_\Mb\circ\alpha_\Mb = \alpha_\Mb$
for all $\eta\in\Aut(\Ff)$ and, if some $\beta$ should have the same
property (replacing $\alpha_\Mb$ by $\beta$) then
$\beta=\alpha_\Mb\circ\gamma$ for a uniquely determined $\gamma$. 
We write $\Af(\Mb)$ for the domain of $\alpha_\Mb$. Next, if 
$\psi:\Mb\to\Nb$, observe that
\[
\eta_\Nb\circ \Ff(\psi)\circ\alpha_\Mb= \Ff(\psi)\circ\eta_\Mb\circ
\alpha_\Mb = \Ff(\psi)\circ \alpha_\Mb\qquad(\eta\in\Aut(\Ff)),
\]
so there is a unique
morphism $\Af(\psi):\Af(\Mb)\to\Af(\Nb)$ such that $\Ff(\psi)\circ\alpha_\Mb = \alpha_\Nb
\circ\Af(\psi)$. 
\begin{Prop} $\Af:\Loco\to\Phys$ is a functor, and the maps $\alpha_\Mb$ 
constitute a subtheory embedding $\alpha:\Af\nto\Ff$. Moreover, if
$\beta:\Bf\nto\Ff$ is any subtheory embedding such that $\eta\circ\beta=\beta$
for all $\eta\in\Aut(\Ff)$, then there is a unique $\hat{\beta}:\Bf\nto\Af$ so
that $\beta=\alpha\circ\hat{\beta}$.
\end{Prop}
\begin{proof} The functorial nature of $\Af$ is justified by the calculations
$
\alpha_\Mb\circ \Af(\id_\Mb) = \Ff(\id_\Mb)\circ\alpha_\Mb = \alpha_\Mb
$
and
\[
\alpha_\Nb\circ\Af(\psi)\circ\Af(\psi') = \Ff(\psi)\circ\Ff(\psi')\circ
\alpha_\Mb = \Ff(\psi\circ\psi')\circ\alpha_\Mb =
\alpha_\Nb\circ\Af(\psi\circ\psi')
\]
together with the monic property of the $\alpha_\Mb$.
By construction the $\alpha_\Mb$ constitute a natural
$\alpha:\Af\nto\Ff$, with the property $\eta\circ\alpha = \alpha$ for
all $\eta\in\Aut(\Ff)$. If $\beta:\Bf\nto\Ff$ with
$\eta\circ\beta=\beta$ for all $\eta\in\Aut(\Ff)$, we take components in
$\Mb$ and use the equalizing property of $\alpha_\Mb$ to deduce
that $\beta_\Mb = \alpha_\Mb\circ \hat{\beta}_\Mb$ for uniquely determined
$\hat{\beta}_\Mb:\Bf(\Mb)\to\Af(\Mb)$. We then calculate
\[
\alpha_\Nb\circ\hat{\beta}_\Nb\circ\Bf(\psi) = \beta_\Nb\circ\Bf(\psi) =
\Ff(\psi)\circ\beta_\Mb = \Ff(\psi)\circ \alpha_\Mb\circ \hat{\beta}_\Mb = 
\alpha_\Nb\circ\Af(\psi)\circ\hat{\beta}_\Mb,
\]
which proves (again, because $\alpha_\Nb$ is monic) that $\hat{\beta}:\Bf\nto\Af$ and
$\beta=\alpha\circ\hat{\beta}$. \end{proof}

The theory $\Af$ is a natural candidate for the theory of observables determined
by the field functor $\Ff$. In the case $\Phys=\Alg$, of course, the algebra $\Af(\Mb)$
may be identified concretely with the subalgebra of $\Ff(\Mb)$ of fixed points under $\eta_\Mb$ ($\eta\in\Aut(\Ff)$). 

However, there are various reasons to be cautious regarding this definition.
First, there is no guarantee that $\Aut(\Af)$ is trivial, 
although this is what one would expect if $\Ff$ is a `reasonable' field functor,
and could be used as a selection criterion for candidate field functors $\Ff$. 
As an example of an `unreasonable' field functor, 
suppose that indeed $\Ff$ is given so that $\Aut(\Af)$ is
trivial, and adopt $\Ff\otimes\Af$ as the field functor. In the
simplest case, $\Aut(\Ff\otimes\Af)=\Aut(\Ff)\otimes\id_\Af$,
and the corresponding
observable functor would be $\Af\otimes\Af$, which has a nontrivial
automorphism corresponding to the interchange of factors. Thus
$\Aut(\Af\otimes\Af)$ has a $\ZZ_2$ subgroup. 

Second, if one applies the same construction to theories defined on possibly
disconnected spacetimes $\Loc$, it can result in `observables'
that are built from `unobservable' elements in different 
spacetime components, whose operational significance is
unclear, to say the least (we will see examples in Sec.~\ref{sect:examples}). 
In these circumstances, 
it is tempting to define the `true algebra of observables' to be
the subalgebra of $\Af(\Mb)$ generated by
the images of $\Af(\iota_{\Mb;\Cb})(\Af(\Cb))$ as $\Cb$ runs
over the connected components of $\Mb$, with canonical inclusions
$\iota_{\Mb;\Cb}:\Cb\to\Mb$.  

Third, in some cases it can happen that the theory $\Af$ is trivial. 
For example, in the case of the classical fields
discussed in Sec.~\ref{sect:examples} there are no nonzero
elements of the symplectic space that are invariant under the action
of all elements of the symmetry group. Similar problems occur
if the theory is then quantized using the Weyl algebra, which has
no fixed points (other than multiples of the unit) under a faithful continuous 
group action if there are no fixed-points in the underlying symplectic space. 

Nonetheless, the above definition is worthy of further investigation and
will turn out to give the expected theory of observables in the scalar
field examples, when quantized using the infinitesimal Weyl algebra.

\section{Energy compactness excludes proper endomorphisms}\label{sect:compactness}

One might suspect that theories admitting endomorphisms that are not automorphisms are unphysical in some way. In this section, we confirm
such suspicions for locally covariant theories whose Minkowski space
versions obey standard assumptions of the Haag--Araki--Kastler framework,
of which the most important will be an energy compactness requirement
weaker than the nuclearity conditions of~\cite{BucWic:1986}. Provided 
such a theory has no `accidental' gauge symmetries in Minkowski space---internal symmetries of the Minkowski space net that do not arise from automorphisms
of the locally covariant theory---then all its endomorphisms are automorphisms.
Moreover, we show that the automorphism group can be given the structure
of a compact topological group. Our argument here is more direct than
standard presentations and uses weaker hypotheses; it is
therefore of independent interest.

Energy compactness conditions were first introduced by Haag and Swieca~\cite{HaaSwi:1965} in an attempt to understand the general
conditions under which a quantum field theory admits a particle interpretation.
A major development in this line of thought occurred with the introduction
of nuclearity criteria by Buchholz and Wichmann~\cite{BucWic:1986}, which
gave more stringent criteria closely linked to the split property (itself
linked to a rich mathematical theory of standard split inclusions of von Neumann algebras~\cite{DopLon:1984}) and good thermodynamic behaviour of the theory~\cite{BucJun:1989}. 
A variety of nuclearity conditions have been proposed
subsequently, see~\cite{BucPor:1990} for a review and~\cite{Bos_uPSC:2005}
for a more recent variant.
The underlying physical idea of all these approaches, arising from the uncertainty principle, is that the number of degrees of freedom available in small phase space volumes should be finite; this cannot be implemented literally, and compactness or its variants often stand in for `finiteness' in the technical conditions imposed. 

Given that endomorphisms of a locally covariant theory are injections and  preserve both localization and  the energy scale,
they intuitively map any given volume of phase space into itself in a volume preserving way. Thus the existence of a proper endomorphism,
(i.e., one that is not an automorphism) can be expected to conflict with
energy compactness on physical grounds, and this is exactly what we will establish.
Note that we are not claiming, nor do we expect, that all models violating energy compactness
admit proper endomorphisms. 

Our result will be proved for locally covariant theory $\langle \Ff,\gamma,\Sf\rangle:\Loc\to\grCAS$ obeying a number of conditions that will now be introduced and discussed. The first assumption collects the basic
conditions required in general spacetimes, while the others relate
specifically to Minkowski space and are largely standard assumptions in
algebraic QFT.
\begin{enumerate}
\item {\em Twisted locality, time-slice, $\Uf$-additivity and local quasi-equivalence} 
Here, $\Uf$ is the usual
faithful functor $\Uf:\grCAS\to \CAlg$, so $\Ff$ is additive as
a theory in $\LCT_\CAlg$.
\item \label{it:unique_inv} {\em Unique Poincar\'e invariant state} In Minkowski space $\Mb_0$, there is a unique state $\omega_0\in \Sf(\Mb_0)$ that is
$\Ff(\psi)$-invariant for all proper orthochronous Poincar\'e transformations $\psi$.
\setcounter{assumptions}{\value{enumi}}
\end{enumerate}
Note that assumption~\eqref{it:unique_inv} does not apply in 
the case of the massless free scalar field, which admits a $1$-parameter
family of Poincar\'e invariant vacuum states in Minkowski space. We will see (albeit not in the $C^*$-setting) that there are no proper endomorphisms of that theory either, but that
the automorphism group is noncompact, in contrast to the situation discussed
in this section.
 
The invariant state $\omega_0$ induces a GNS representation $(\Hc,\pi,\Omega)$ of the theory in Minkowski space, and hence a local net of
$C^*$-algebras $\Ffr(O):=\pi(\Ff^\kin(\Mb_0;O))$ on $\Hc$ indexed by 
relatively compact, connected and nonempty $O\in \OO(\Mb_0)$. Taking
double-commutants we obtain a net of von Neumann algebras $\Mfr(O):=
\Ffr(O)''$, with the same index set. The GNS representation and local
nets are assumed to obey a number
of standard conditions:
\begin{enumerate}\setcounter{enumi}{\value{assumptions}}
\item {\em Faithfulness, irreducibility and separability} The GNS representation $\pi$ is a faithful and irreducible representation of $\Ff(\Mb_0)$ on a separable Hilbert space $\Hc$.
\item {\em Covariance and spectrum condition} (a) The algebra automorphisms of $\Ff(\Mb_0)$ induced by 
the proper orthochronous Poincar\'e group can be unitarily implemented in $(\HH,\pi,\Omega)$ by a 
strongly continuous unitary representation $\Lambda\mapsto U(\Lambda)$ so that $U(\Lambda)\Omega=\Omega$ and
\[
U(\Lambda)\Ffr(O)U(\Lambda)^{-1} = \Ffr(\Lambda O);
\]
(b) the self-adjoint generators $P_\mu$ of the translation subgroup have joint spectrum contained
in the forward light-cone. 
\item {\em Reeh--Schlieder}\label{it:ReehSchlieder} For all nonempty $O\in\OO(\Mb_0)$, the subspace $\Ffr(O)\Omega$ is dense in $\Hc$.
\item {\em Energy compactness}\label{it:compactness} For some nonempty $O\in\OO(\Mb_0)$ and $\beta>0$, the set
\[
\cN = \{e^{-\beta H} W\Omega:~W\in\Mfr(O)~\text{s.t.}~W^*W=\II\}
\]
is a relatively compact subset of $\Hc$ (with necessarily dense linear span, by the Reeh--Schlieder condition and because $e^{-\beta H}=(e^{-\beta H})^*$ has trivial kernel), where $H=P_0$ is the Hamiltonian with 
respect to some system of inertial coordinates. 
\setcounter{assumptions}{\value{enumi}}
\end{enumerate}
The last of the standard assumptions, {\em twisted duality}~\cite{DHRi},
requires some notation. Let $\Gamma$ be the unitary implementing $\gamma_{\Mb_0}$,
so $\Gamma^2 = \II$, $\Gamma\Omega= \Omega$ and $\Gamma\pi(A)\Gamma^{-1}=  \pi(\gamma_{\Mb_0}(A))$ for all $A\in\Ff(\Mb_0)$, and define a unitary
\[
Z = \frac{1-i}{2} + \frac{1+i}{2}\Gamma.
\]
For any subset of bounded operators $\Mfr\subset\Bfr(\Hc)$, let $\Mfr^t=Z\Mfr Z^{-1}$ and write $\Mfr^t(O)=\Mfr(O)^t$. If $\gamma_{\Mb_0}(B)=(-1)^\sigma B$ ($\sigma=\{0,1\}$) then
\[
[\pi(A),Z\pi(B)Z^{-1}] = \pi(\bbLbrack A, B \bbRbrack) (-i\Gamma)^{\sigma} 
\]
for any $A\in\Ff(\Mb_0)$, from which the expression for general $B$ may be obtained by linearity. Twisted locality of $\langle\Ff,\gamma,\Sf\rangle$ implies that $\Ffr(O)$ and $\Ffr(\tilde{O})^t$ commute for any spacelike separated $O,\tilde{O}\in\OO(\Mb_0)$, and hence 
$\Mfr(O) \subset \Mfr^t(\tilde{O})'$. Twisted duality is the following more specific statement.
\begin{enumerate}\setcounter{enumi}{\value{assumptions}}
\item {\em Twisted duality}\footnote{In the literature, proofs of (twisted) Haag duality for particular models in Minkowski space typically apply to Cauchy developments of sufficiently well-behaved subsets of
constant-time hypersurfaces, e.g., double-cones. While one might expect the same to be true for general diamonds, we err on the side of
caution by allowing the possibility that twisted duality might only
be known for a special class of diamonds, which is nonetheless
large enough to generate any local algebra.} \label{it:twduality} 
There is a subset $\KK_\diamond\subset\OO(\Mb_0)$ such that 
every $\tilde{O}\subset\KK_\diamond$ is a diamond (i.e., a multi-diamond
with one connected component) and  (a) for all nonempty
relatively compact and connected $O\in\OO(\Mb_0)$,
\[
\Ffr(O) = \bigvee_{\tilde{O}\subset O} \Ffr(\tilde{O}),
\]
where the $C^*$-algebraic join is taken over $\tilde{O}\in\KK_\diamond$
contained in $O$,
and (b) for every $O\in\KK_\diamond$, we have
\[
\Mfr(O) = \bigcap_{\tilde{O}\subset O'} \Mfr^t(\tilde{O})',
\]
where the intersection runs over all $\tilde{O}\in\KK_\diamond$ contained in the causal complement $O'=\Mb\setminus \cl J_\Mb(O)$ of $O$. 
\setcounter{assumptions}{\value{enumi}}
\end{enumerate}

As already mentioned, assumptions~(1)--(\ref{it:twduality}) are standard conditions in Minkowski space AQFT and are satisfied, for example by
models of free scalar fields (cf. Sec.~\ref{sec:Weyl}).  Indeed, the version of twisted duality stated here is slightly weaker than that in~\cite{DHRi} and the crucial assumption on
energy compactness is much weaker than the nuclearity condition of~\cite{BucWic:1986}, which would require $\cN$ to be a nuclear subset of $\Hc$, with nuclearity index obeying prescribed bounds in terms of the size of $O$ and the inverse temperature $\beta$. Here, our condition is not required to hold for all $\beta$ or $O$ and therefore can incorporate some theories with a maximum temperature. The exponential energy damping is not critical. One could work just as well
with a spectral projection of $H$, as in the Haag--Swieca criterion~\cite{HaaSwi:1965}; again, our condition would be
weaker, because Haag and Swieca also impose conditions on the `approximate dimension' of the sets they consider. The utility of energy compactness
arises from the next result, which is proved at the end
of this section. 
\begin{Lem}\label{lem:iso_to_unitary}
Suppose $\cN$ is a relatively compact subset of a separable Hilbert space $\Hc$, with dense linear span. (a)  If $T\in\Bfr(\Hc)$ is
an isometry with $T\cN\subset \cN$, then $T$ is unitary, and $T\cl \cN  =\cl\cN$.
(b) Let $G$ be the group of unitary operators $U\in\Bfr(\Hc)$ obeying $U\cl\cN = \cl\cN$. Then $G$ is compact with respect to
the strong operator topology. 
\end{Lem}
This result permits us to give an apparently new proof of the compactness of the maximal global gauge group $G_{\rm max}$, defined by Doplicher, Haag and Roberts~\cite{DHRi}, consisting of all unitary operators $U$ on $\Hc$ that commute with the representation of the Poincar\'e group, preserve the vacuum vector and act strictly locally on the net of local von Neumann algebras, in 
the sense that  $U\Mfr(O) U^{-1} = \Mfr(O)$
for all relatively compact connected nonempty $O\in\OO(\Mb_0)$. Compactness of $G_{\rm max}$ has been proved under various assumptions in the past,
e.g., the existence of an asymptotically complete scattering theory with finite particle multiplets~\cite{DHRi}, or under the assumption of nuclearity, which implies the split property~\cite{BucWic:1986} and
hence compactness of $G_{\rm max}$ by (the proof of)~\cite[Theorem~10.4]{DopLon:1984}. However,
we wish to point out that compactness of $G_{\rm max}$ may be established directly and under the weaker energy compactness condition assumed here,
as a consequence of Lemma~\ref{lem:iso_to_unitary}(b). (Twisted duality
is not needed for this argument.)
\begin{Prop} \label{prop:Gmax} Under assumptions (1)--(\ref{it:compactness}),  
the group $G_{\rm max}$ is compact in the strong operator topology,
and $\Aut(\langle \Ff,\gamma,\Sf\rangle)$ is compact in the 
natural weak topology.
\end{Prop}
\begin{proof} Any $U\in G_{\rm max}$ preserves the relatively compact set
$\cN = e^{-\beta H} \Mfr(O)\Omega$, because $U$ acts strictly locally,
commutes with the Hamiltonian and obeys $U\Omega=\Omega$. Combining
both parts of
Lemma~\ref{lem:iso_to_unitary},  $G_{\rm max}$ is therefore contained in a group of unitaries
that is compact in the strong operator topology. As $G_{\rm max}$ is
closed in this topology, because its defining relations are preserved under strong limits~\cite{DHRi}, it is compact. Hence, the 
$\omega_0$-topology on $\Aut(\langle \Ff,\gamma,\Sf\rangle)$ [defined
at the end of Sec.~\ref{sec:gauge_gp}] is
compact, and is equivalent to the natural weak topology by Theorem~\ref{thm:tops} of Appendix~\ref{appx:topology}.
\end{proof}

Returning to the main theme of excluding proper endomorphisms, 
our first task is to show that any endomorphism is $\langle \Ff,\gamma,\Sf\rangle$ 
is indistinguishable, in Minkowski space, from a gauge transformation in $G_{\rm max}$. In what follows, we abuse notation slightly by using the same symbol for both an endomorphism of $\langle \Ff,\gamma,\Sf\rangle$ and the underlying endomorphism of $\Ff$.
Our first result extends the unitary
representation of $\Aut(\langle \Ff,\gamma,\Sf\rangle)$ from Proposition~\ref{prop:urep}.
\begin{Thm} \label{thm:rho}
Under assumptions (1)--(\ref{it:twduality}),
there is a faithful homomorphism of monoids $\rho:\End(\langle \Ff,\gamma,\Sf\rangle)\to G_{\rm max}$, 
 obeying 
\begin{equation}\label{eq:rho_def}
\rho(\eta)\pi(A)\Omega = \pi(\eta_{\Mb_0}A)\Omega \qquad (\eta\in\End(\langle \Ff,\gamma,\Sf\rangle),~A\in\FF(\Mb_0));
\end{equation}
in particular, each $\rho(\eta)$ unitarily implements $\eta_{\Mb_0}$.
\end{Thm}
\begin{proof} Let $\eta\in\End(\langle \Ff,\gamma,\Sf\rangle)$.
By definition of the category $\LCT_\grCAS$, we have $\eta\circ\gamma=\gamma\circ\eta$ and $\eta^*\Sf\subset\Sf$.
As $\omega_0$ is the unique Poincar\'e invariant state, Lemma~\ref{lem:invariant_states} entails that
$\eta_{\Mb_0}^*\omega_0=\omega_0$. Then the calculation, for arbitrary $A\in\Ff(\Mb_0)$, 
\[
\|\pi(\eta_{\Mb_0} A)\Omega\|^2 = \omega_0((\eta_{\Mb_0} A)^*\eta_{\Mb_0}(A))
= \omega_0(A^*A) = \|\pi(A)\Omega\| ^2
\]
shows that the equation $T\pi(A)\Omega = \pi(\eta_{\Mb_0} A)\Omega$ (for all $A\in\Ff(\Mb_0)$)
defines $T$ unambiguously on a dense domain in the GNS Hilbert space, and extends
by continuity to define an isometry of $\Hc$ into itself.

Next, let $\tau_t:(x^0,\xb)\mapsto (x^0+t,\xb)$ be the time translation automorphism of $\Mb_0$ in some system of standard inertial coordinates,  unitarily implemented so that
$
e^{iHt} \pi(A)\Omega = \pi(\Ff(\tau_t)A)\Omega$,
where $H$ is the Hamiltonian in these coordinates. Because $\eta$ is natural, we must have
$\eta_{\Mb_0}\circ \Ff(\tau_t) = \Ff(\tau_t)\circ\eta_{\Mb_0}$, which gives
\[
Te^{iHt}\pi(A)\Omega = \pi(\eta_{\Mb_0}\circ \Ff(\tau_t) A)\Omega = 
 \pi(\Ff(\tau_t) \circ \eta_{\Mb_0}A)\Omega = e^{iHt} T\pi(A)\Omega 
\]
i.e., the isometry $T$ commutes with $e^{iHt}$ on a dense domain and hence all of $\Hc$.
We deduce that $T$ also commutes with $e^{-\beta H}$ for any $\beta \ge 0$.\footnote{
Take any Schwartz test function $f$ with $\hat{f}(\lambda)=e^{-\beta\lambda}$ for $\lambda\ge 0$;
then $\int_\RR f(t)\ip{\varphi}{e^{-iHt}\psi}dt =\ip{\varphi}{e^{-\beta H}\psi}$ for any $\varphi,\psi\in\Hc$, 
from which  $Te^{-\beta H}=e^{-\beta H}T$ follows (cf.\ the proof of Theorem~VIII.13 in~\cite{ReedSimon:vol1}).} 

For any unital $C^*$-algebra $\Ac$, we write $\Ac_{(1)}=\{A\in\Ac:~A^*A=\II\}$. 
According to the energy compactness assumption, we may choose nonempty $O\in\OO(\Mb_0)$ and $\beta>0$ so that $\cN=e^{-\beta H}\Ffr_{(1)}(O)\Omega$ 
is a subset of the relatively compact set $e^{-\beta H}\Mfr_{(1)}(O)\Omega$.
Hence $\cN$ is relatively compact, with dense linear span. 
Moreover,
\begin{align*}
T\cN &= e^{-\beta H} T\pi(\Ff^\kin_{(1)}(\Mb_0;O))\Omega
=  e^{-\beta H} \pi(\eta_{\Mb_0}\Ff^\kin_{(1)}(\Mb_0;O))\Omega \\
&\subset e^{-\beta H} \pi(\Ff^\kin_{(1)}(\Mb_0;O))\Omega = \cN,
\end{align*}
and  Lemma~\ref{lem:iso_to_unitary}(a) entails that $T$ is unitary. 
Consequently, considering the 
action on the dense domain $T\pi(\Ff(\Mb_0))\Omega$, we see that $T$ implements $\eta_{\Mb_0}$, i.e., $T\pi(A)T^{-1}=\pi(\eta_{\Mb_0}A)$. Therefore
\[
T\Ffr(O)T^{-1} = \pi(\eta_{\Mb_0}(\Ff^\kin(\Mb_0;O))) \subset \pi(\Ff^\kin(\Mb_0;O)) =\Ffr(O)
\]
for all nonempty $O\in\OO(\Mb_0)$. Were $T^{-1}$ known to implement
an endomorphism of $\langle\Ff,\gamma,\Sf\rangle$ (as would be the case if $\eta$ was an automorphism), the above inclusion would become an equality.\footnote{There are cases where proper endomorphisms
of algebras are unitarily implemented, e.g., shrinking scale transformation on a suitable local algebra
in a conformally covariant theory.} Lacking such knowledge, however, we
proceed using twisted duality.

Passing to the local von Neumann algebras $\Mfr(O)$, the argument so far has established that $T\Mfr(O)T^{-1}\subset \Mfr(O)$, 
for all relatively compact connected nonempty $O\in\OO(\Mb_0)$. As $\eta$ commutes with $\gamma$, we have 
$T\Gamma=\Gamma T$, and therefore $TZ=ZT$. In particular, $T\Mfr^t(O)T^{-1} = (T\Mfr(O)T^{-1})^t\subset \Mfr^t(O)$. Thus, if $\tilde{O}$ is spacelike to $O$,
\[
[T^{-1}\Mfr(O)T,\Mfr^t(\tilde{O})] = T^{-1}[\Mfr(O),T\Mfr^t(\tilde{O})T^{-1}] T
\subset T^{-1}[\Mfr(O),\Mfr^t(\tilde{O})] T = \{0\}
\]
and by twisted duality, we see that $T^{-1}\Mfr(O)T\subset \Mfr(O)$
for all $O\in\KK_\diamond$.  Putting this together
with our earlier result gives $T\Mfr(O)T^{-1} = \Mfr(O)$ for all $O\in\KK_\diamond$. Now if $O\in\OO(\Mb_0)$ is nonempty, connected
and relatively compact, the additivity assumption~(\ref{it:twduality}a)
yields
\[
\Ffr(O) =\bigvee_{\tilde{O}\subset O}\Ffr(\tilde{O})
\qquad\text{and hence}\qquad
\Mfr(O) =\bigvee_{\tilde{O}\subset O}\Mfr(\tilde{O}),
\]
where the joins run over $\tilde{O}\in\KK_\diamond$; the first is
taken in the sense of $C^*$-algebras, while the second is taken in the
von Neumann sense and follows on taking weak closures. It follows immediately that $T\Mfr(O)T^{-1} = \Mfr(O)$ for all such $O$. 

Summarising: $T$ is a unitary operator, commuting with the unitary representation of the proper, orthochronous Poincar\'e group, acting strictly locally on the net and preserving $\Omega$. Setting
$\rho(\eta)=T$, we have a map $\rho:\End(\langle \Ff,\gamma,\Sf\rangle)\to G_{\rm max}$ obeying Eq.~\eqref{eq:rho_def}
and hence $\rho(\eta)\pi(A)\rho(\eta)^{-1} = \pi(\eta_{\Mb_0}A)$. 
It is evident from these equations that $\rho$ is a homomorphism and obeys $\rho(\id_{\langle \Ff,\gamma,\Sf\rangle})=\II_\HH$. 
Further, as $\pi$ is faithful, $\rho(\eta)=\II$ if and only if $\eta_{\Mb_0}=\id_{\Ff(\Mb_0)}$. By Theorem~\ref{thm:reduce_to_Mink}
we have $\eta=\id_{\langle \Ff,\gamma,\Sf\rangle}$, which 
completes the proof that $\rho$ is a faithful monoid homomorphism. \end{proof}

To say more, we require an additional assumption. 
\begin{enumerate}\setcounter{enumi}{\value{assumptions}}
\item {\em Implementation of $\rce$}\label{it:rce} The relative Cauchy evolution can be unitarily implemented in the
GNS representation of $\omega_0$, i.e., to each $\hb\in H(\Mb_0)$,
there is a unitary $V(\hb)$ on $\Hc$, such that
\[
V(\hb)\pi(A)V(\hb)^{-1} = \pi(\rce_\Mb[\hb] A) \qquad (A\in\Ff(\Mb)).
\]
We assume in addition that $\ip{\Omega}{V(\hb)\Omega}\neq 0$ for
all $\hb\in H(\Mb_0)$. 
\setcounter{assumptions}{\value{enumi}}
\end{enumerate} 
Note that $V(\hb)$ cannot leave the vacuum invariant (unless the relative Cauchy
evolution is trivial) by Proposition~\ref{prop:localV} below. 
The assumption that $\ip{\Omega}{V(\hb)\Omega}\neq 0$ 
is motivated by the idea that $V(\hb)\Omega$ is a `squeezed' vacuum.
This can be seen explicitly for free Bose theories, where the relative
Cauchy evolution of the QFT is a Bogoliubov transformation induced by  
the relative Cauchy evolution in the underlying classical theory.
(One might
suspect that it actually follows from the other axioms, but we do not have a proof of this at present.) In fact, the assumption should hold generally
at least for $\hb$ in a neighbourhood of the zero test tensor, if the
relative Cauchy evolution is assumed to be continuous at $\hb=\Ob$,
which in turn is a precondition for the existence of a stress-energy
tensor as a field in the GNS representation of $\omega_0$. The utility of this
condition is made clear in the following result. 
\begin{Prop} \label{prop:UVprime}
Under assumptions (1)-(\ref{it:rce}),
the homomorphism $\rho$ of Theorem~\ref{thm:rho} obeys
\[
\Im \rho\subset G_\rce := G_{\rm max}\cap\{V(\hb):\hb\in H(\Mb_0)\}'.
\]
The group $G_\rce$ is compact in the strong operator topology of $\HH$.  
\end{Prop}
\begin{proof} Let $\eta\in\End(\langle \Ff,\gamma,\Sf\rangle)$. 
For any $A\in\Ff(\Mb_0)$ and $\hb\in H(\Mb_0)$, the intertwining property~\eqref{eq:intertwine} entails
\begin{align*}
\rho(\eta)V(\hb)\pi(A)V(\hb)^{-1}\rho(\eta)^{-1} &= \pi(\eta_{\Mb_0}\circ\rce_{\Mb_0}[\hb] A) = 
\pi(\rce_{\Mb_0}[\hb]\circ \eta_{\Mb_0} A) \\ &= V(\hb)\rho(\eta)\pi(A) \rho(\eta)^{-1}V(\hb)^{-1}.
\end{align*}
As $\pi$ is irreducible, this implies that $\rho(\eta)V(\hb) = \alpha_\eta(\hb)V(\hb)\rho(\eta)$, for
some constant $\alpha_\eta(\hb)$, necessarily of unit modulus. But $\rho(\eta)\Omega=\Omega=\rho(\eta)^*\Omega$ 
and $\ip{\Omega}{V(\hb)\Omega}\neq 0$, so
$\alpha_\eta(\hb)=1$. Thus $\Im\rho\subset G_{\rce}$; as the
defining relations of this group are preserved under strong limits,
it is a closed subset of $G_{\rm max}$ and hence compact, by Proposition~\ref{prop:Gmax}.
 \end{proof} 

As a minor digression, we also note the following application of  Proposition~\ref{prop:UVprime}
(from which the affiliation of the stress-energy tensor to the local net would follow). 
\begin{Prop} \label{prop:localV}
 Subject to assumptions (1)-(\ref{it:rce}) above,
if $\hb\in H(\Mb_0;O)$ for some $O\in\KK_\diamond$, then $V(\hb)\in\Mfr(O)$.  
Moreover, $V(\hb)\Omega\in\CC\Omega$ only if $\rce_{\Mb_0}[\hb]$ is the identity automorphism.
\end{Prop}
\begin{proof} If $\tilde{O}\in\OO(\Mb_0)$ is any
relatively compact and connected subset of the causal complement $O'=
\Mb_0\setminus \cl J_{\Mb_0}(O)$ of $O$, then $\rce_\Mb[\hb]A=A$ for all $A\in \Ff^\kin(\Mb;\tilde{O})$
by Proposition~3.5 of~\cite{FewVer:dynloc_theory}, and hence $V(\hb)\in \Ffr(\tilde{O})'
=\Mfr(\tilde{O})'$. Applying Proposition~\ref{prop:UVprime} to $\gamma$, we see that 
$V(\hb)$ commutes with $\Gamma$ and hence $Z$, so we also have $V(\hb)\in \Mfr^t(\tilde{O})'$.
The result $V(\hb)\in\Mfr(O)$ for $O\in\KK_\diamond$ now follows from twisted duality and arbitrariness of $\tilde{O}$.
If $V(\hb)\Omega=\alpha\Omega$ ($\alpha\in\CC$), then $V(\hb)=\alpha\II$ as
$\Omega$ is a separating vector for $\Mfr(O)$ by twisted locality and the Reeh--Schlieder property. 
As $\pi$ is faithful, this implies $\rce_{\Mb_0}[\hb]$ is trivial.
\end{proof}

Subject to the assumptions made so far, we have shown that any endomorphism of $\langle \Ff,\gamma,\Sf\rangle$ acts as an {\em automorphism} of the Minkowski space net. Provided
that all such automorphisms are related to symmetries of the 
full theory, i.e., there are no `accidental symmetries', we may then exclude proper endomorphisms. Thus, our final assumption is:
\begin{enumerate}\setcounter{enumi}{\value{assumptions}}
\item {\em Absence of accidental gauge symmetries} To each $U\in G_\rce$ 
there is an automorphism $\zeta(U)\in\Aut(\langle \Ff,\gamma,\Sf\rangle)$ such that $\rho(\zeta(U))=U$. \label{it:final}
\setcounter{assumptions}{\value{enumi}}
\end{enumerate} 
 
\begin{Thm} \label{thm:end_aut}
For a theory $\langle\FF,\gamma,\Sf\rangle\in\LCT_\grCAS$ obeying
assumptions (1)--(\ref{it:final}) above, 
\[
\End(\langle \Ff,\gamma,\Sf\rangle)=\Aut(\langle \Ff,\gamma,\Sf\rangle)\cong G_\rce
\]
(an isomorphism of groups). Moreover, $G_{\rm max}$ and $G_\rce$ are compact in the strong operator topology, while $\Aut(\langle \Ff,\gamma,\Sf\rangle)$ is compact in the natural weak topology given in Sec.~\ref{sec:gauge_gp}.
\end{Thm}
\begin{proof} Let $\eta\in \End(\langle \Ff,\gamma,\Sf\rangle)$.
Proposition~\ref{prop:UVprime} and assumption~(\ref{it:final}) yield
$\rho(\eta)=\rho(\zeta(\rho(\eta))$; faithfulness
of $\pi$ gives $\eta_{\Mb_0} = \zeta(\rho(\eta))_{\Mb_0}$
and hence $\eta = \zeta(\rho(\eta))$ by Theorem~\ref{thm:reduce_to_Mink}. 
Thus $\End(\langle \Ff,\gamma,\Sf\rangle)=\Aut(\langle \Ff,\gamma,\Sf\rangle)$ and the monoid homomorphism $\rho$, already
known to be faithful, is now bijective and becomes a group
isomorphism. The compactness statements were proved 
in Propositions~\ref{prop:Gmax} and~\ref{prop:UVprime}. 
\end{proof}

Some remarks are in order. First, it would be interesting derive
the content of assumption~(\ref{it:final}) from the other axioms or another more primitive requirement. We leave this
as an open problem. Alternatively, one could drop this assumption and read our results
as proving that proper endomorphisms are associated with
accidental symmetries of the Minkowski net. Second, in general 
categories, objects admitting no proper monic endomorphisms are called {\em Dedekind finite}: for the category of sets these are indeed
the finite sets, while for the category of vector spaces 
they are the finite-dimensional spaces (see~\cite{Stout:1987} for other examples of Dedekind finiteness in different categories, including 
a number of salutary counterexamples). Thus  theories obeying our
assumptions are Dedekind-finite objects in the category of locally covariant theories.  An immediate consequence is a strong version of the 
Schr\"oder--Bernstein property for locally covariant theories.  
\begin{Cor}
Suppose $\langle \Af,\gamma,\Sf\rangle$ and $\langle \Bf,\delta,\Tf\rangle$ are theories in $\LCT_\grCAS$,
at least one of which obeys assumptions (1)--(\ref{it:final}). If there are 
subtheory embeddings $\alpha:\langle \Af,\gamma,\Sf\rangle\nto \langle \Bf,\delta,\Tf\rangle$
and $\beta:\langle \Bf,\delta,\Tf\rangle\nto \langle \Af,\gamma,\Sf\rangle$
then both $\alpha$ and $\beta$ are  isomorphisms. 
\end{Cor}
\begin{proof} Without loss, suppose that $\langle \Af,\gamma,\Sf\rangle$ obeys the
assumptions (1)--(\ref{it:final}). Then
$\eta=\beta\circ\alpha$ is an endomorphism of $\langle \Af,\gamma,\Sf\rangle$, 
and hence an automorphism by Theorem~\ref{thm:end_aut}. As $\beta$ is
monic, we may deduce that both $\beta$ and $\alpha$ are isomorphisms.\footnote{
Evidently,
$\beta\circ(\alpha\circ\eta^{-1})=\id_{\langle \Af,\gamma,\Sf\rangle}$, and $\beta\circ (\alpha\circ\eta^{-1})\circ\beta = 
\beta=\beta\circ\id_{\langle \Bf,\delta,\Tf\rangle}$, entailing that $(\alpha\circ\eta^{-1})\circ\beta=\id_\Bf$ 
using the monic property of $\beta$. Hence $\beta$ is invertible, and so is $\alpha=\beta^{-1}\circ\eta$.}
\end{proof}

Finally, we comment on an alternative (though more restrictive) energy compactness condition for Minkowski space theories, known as the microscopic phase space condition~\cite{Bos_uPSC:2005}.  Theories obeying this condition, along with
other standard assumptions, have a definable field content forming finite dimensional subspaces 
that are sufficient to describe the theory at different orders of a short-distance approximation at sufficiently low energies, and reproduce the total field content of~\cite{FreHer:1981}. This might offer an alternative
approach to the results of this section that avoids the need for twisted locality and twisted duality.

It remains to prove the Hilbert space result used in our argument.

\begin{proofof}{Lemma~\ref{lem:iso_to_unitary}}  (a) For any $\psi\in\cN$, the sequence $T^k\psi$ is contained in
$\cN$ and therefore has a subsequence $T^{k_r}\psi$ converging in $\Hc$. In particular, 
for any $\epsilon>0$ there is $R>0$ such that $\|(T^{k_s}-T^{k_r})\psi\|<\epsilon$ for
all $s> r>R$. But $T$ is an isometry, so we also have $\|T^{k_s-k_r}\psi - \psi\|<\epsilon$
and we may deduce the existence of a sequence $j_r\to\infty$ with $T^{j_r}\psi\to\psi$. 
Hence $\cN\subset \cl(T\cN)$, which, together with $T\cN\subset \cN$,
implies $\cl(T\cN)= \cl(\cN)$. Moreover, as
$T^{(j_r-1)}\psi= T^*T^{j_r}\psi \to T^*\psi$, we have $\|T^*\psi\| =
\lim_{r} \|T^{j_r-1}\psi\| =  \|\psi\|$ for all $\psi\in\cN$. Now $\II-TT^*$ is a projection, so the elementary identity
$\| (\II-TT^*)\psi\|^2 = \ip{\psi}{(\II-TT^*)\psi} = \|\psi\|^2 -\|T^*\psi\|^2 = 0$ gives $TT^*\psi=\psi$ for all $\psi\in\cN$. As $\cN$ has dense linear span, we conclude that $T$ is unitary.
Accordingly a sequence in $\HH$ is Cauchy if and only if its image under $T$ is, from which  
we may obtain $T\cl\cN = \cl T\cN = \cl\cN$. 

(b) 
As $\Hc$ is separable, the strong operator topology is metrisable on any bounded subset of $\Bfr(\Hc)$
\cite[Proposition~2.7]{Takesaki_vol1} and convergence and compactness can be determined sequentially. Now, $G$ is certainly closed: any strong limit of a sequence in $G$ must be an isometry mapping the compact set $\cl\cN$ into itself; by Lemma~\ref{lem:iso_to_unitary}(a)
the limit is therefore contained in $G$. Turning to compactness, choose a sequence $U_n$ in $G$, and fix a countable linearly independent set of vectors $\psi_j\in\cN$ with dense linear span. The sequence $U_n\psi_1$ in $\cl\cN$ must have convergent subsequences, and we choose a  subsequence $U_{n^{(1)}(r)}$ so that
$U_{n^{(1)}(r)}\psi_1$ converges. Proceeding inductively, we choose successive subsequences $U_{n^{(k)}(r)}$ so that $U_{n^{(k)}(r)}\psi_j$ converges for each $1\le j\le k$. The diagonal
subsequence $V_k = U_{n^{(k)}(k)}$ converges strongly on each $\psi_j$
and hence on any finite linear combination of the $\psi_j$. As $\Span\{\psi_j\}$ is dense and the $V_k$ are uniformly bounded,  it follows that the sequence $V_k$ converges strongly to a limit in $G$. 
Thus $G$ is sequentially compact and hence compact. 
\end{proofof}

\section{Scalar fields}\label{sect:examples}

As a concrete example, we consider theories of finitely many free
minimally coupled scalar fields, specified by a finite set $\Mc\subset [0,\infty)$ representing the mass spectrum and a function
$\nu:\Mc\to\NN$ giving the number of field species with each mass
(NB $\nu(m)>0$ for each $m\in\Mc$). The total number of
species is denoted $|\nu|=\sum_{m\in\Mc}\nu(m)$. We consider
both the  classical and quantum theories, working in a fixed spacetime dimension $n\ge 2$, though some of our results on the quantum field theory hold only in dimension $n\ge 3$. 

\subsection{Classical theory}

The classical theory is described using  the category of complexified symplectic spaces,
$\Sympl$, objects of which are triples
$(V,\Gamma,\sigma)$, consisting of a complex
vector space $V$, an antilinear conjugation $\Gamma:V\to V$ and
a weakly nondegenerate antisymmetric bilinear form $\sigma:V\times V\to\CC$
such that $\sigma(\Gamma v,\Gamma w) = \overline{\sigma(v,w)}$ for all $v,w\in V$.
In our applications, $V$ will be a
space of complex-valued functions and $\Gamma$ will be induced by complex conjugation: $(\Gamma f)(p)=\overline{f(p)}$.
A morphism $S:(V,\Gamma,\sigma)\to (V',\Gamma',\sigma')$ in $\Sympl$ is a 
(necessarily injective) complex linear
map $S:V\to V'$, such that $\sigma'(Sv,Sw) = \sigma(v,w)$ and $\Gamma' Sv =
S\Gamma v$ hold for all $v,w\in V$. Relaxing the requirement for the bilinear
forms to be nondegenerate, we obtain the category of complexified
presymplectic spaces $\preSympl$. There is an obvious forgetful
functor $\Uf:\Sympl\to\preSympl$.

The theory of a single minimally coupled field of mass $m\ge 0$ is 
standard and was described in detail from the functorial perspective in~\cite{FewVer:dynloc2}. 
We recall the essential facts only. For any $\Mb\in\Loc$, let $\Sol_m(\Mb)$ be the 
space of complex-valued solutions $\phi$ to
\begin{equation}\label{eq:KG}
(\Box_\Mb+m^2)\phi=0
\end{equation} that have compact support on Cauchy surfaces in
$\Mb$.\footnote{Our notation differs slightly from that of ~\cite{FewVer:dynloc2}, where the mass
was not indicated explicitly; in particular, for $m=0$ the present $\Sol_0(\Mb)$ does {\em not} coincide with the space $\Sol_0(\Mb)$ studied in~\cite{FewVer:dynloc2}, which permits solutions that are compactly supported on Cauchy surfaces following modification by a locally constant function. See also Sec.~\ref{sect:observables} here.} Equipping $\Sol_m(\Mb)$ with complex conjugation and the antisymmetric bilinear
form 
\[
\sigma_{m,\Mb}(\phi,\phi') = \int_\Sigma \left(\phi n^a\nabla_a\phi' - \phi'
n^a\nabla_a\phi\right) d\Sigma,
\]
where $\Sigma$ is a Cauchy surface with future-pointing unit normal $n^a$, 
$\Sol_m(\Mb)$ becomes a complexified symplectic space (independent 
of the choice of $\Sigma$). The space $\Sol_m(\Mb)$
can also be expressed as $\Sol_m(\Mb)=E_{m,\Mb}\CoinX{\Mb}$,
where $E_{m,\Mb}$ is the solution operator obtained as the difference
of the advanced and retarded Green functions for \eqref{eq:KG}. 
If $\psi:\Mb\to\Nb$ in $\Loc$, the push-forward $\psi_*:\CoinX{\Mb}\to\CoinX{\Nb}$, defined in Eq.~\eqref{eq:psistar}, 
induces a unique linear map  $\Sol_m(\psi):\Sol_m(\Mb)\to \Sol_m(\Nb)$ 
such that $\Sol_m(\psi)E_{m,\Mb}f = E_{m,\Nb}\psi_*f$; this is in fact
a $\Sympl$-morphism and makes $\Sol_m:\Loc\to\Sympl$ a functor.

The full theory is described by a functor $\Sol:\Loc\to\Sympl$ with
\[
\Sol(\Mb) = \bigoplus_{m\in \Mc} \Sol_m(\Mb)\otimes \CC^{\nu(m)} ,
\]
equipped with complex conjugation and antisymmetric form
\[
\sigma_\Mb((\phi_m\otimes z_m)_{m\in\Mc},(\phi'_m\otimes
z'_m)_{m\in\Mc}) = \sum_{m\in\Mc} \sigma_m(\phi_m,\phi'_m)z_m^T z'_m,
\]
where each $\phi_m\in \Sol_m(\Mb)$, $z_m\in \CC^{\nu(m)}$; here, we regard 
$\CC^{k}$ as a space of column vectors and write $T$ for transpose. 

The theory $\Sol$ inherits the timeslice property from the theories $\Sol_m$~\cite{FewVer:dynloc2} and therefore has a relative Cauchy evolution, which is differentiable in the weak symplectic 
topology (see~\cite{FewVer:dynloc2}).  Let $\Sym(\Mb)$ denote the space of smooth symmetric second rank covariant tensor fields of compact support on each $\Mb\in\Loc$. Then, for each $\Mb\in\Loc$ and $\fb\in\Sym(\Mb)$, there
exists a linear and symplectically skew-adjoint map $F_\Mb[\fb]$ on $\Sol(\Mb)$ such that
\begin{equation}\label{eq:wsdiff}
\sigma_\Mb(F_\Mb[\fb]\phi,\phi')= \left.\frac{d}{ds} \sigma_\Mb(\rce_\Mb[s\fb]\phi,\phi')\right|_{s=0}\qquad (\phi\in\Sol(\Mb)).
\end{equation}
The maps $F_\Mb[\fb]$ are related to the classical stress-energy tensor: in fact, 
\begin{equation}
\sigma_\Mb(F_\Mb[\fb]\phi, \overline{\phi}) = \int f_{ab} T_\Mb^{ab}[\phi]\,\dvol_\Mb ,
\end{equation}
where $\Tb_\Mb[\phi]$ is the stress-energy tensor of $\phi\in\Sol(\Mb)$. 
In view of the intertwining property \eqref{eq:intertwine}, any endomorphism $\eta\in\Sol$ obeys $F_\Mb[\fb]\eta_\Mb = \eta_\Mb F_\Mb[\fb]$ and indeed
\[
\Tb_\Mb[\eta_\Mb\phi] = \Tb_\Mb[\phi] \qquad (\phi\in\Sol(\Mb)).
\]

Although $\Sympl$ does not admit general categorical unions,  $\preSympl$ does, namely  the linear span of (conjugation-invariant) subspaces. Moreover, 
$\preSympl$ also has equalizers; the equalizer of two morphisms being the
inclusion morphism of the subspace on which they agree (i.e., the kernel of their difference) in their common domain. Let $\Mb\in\Loc$ be
any spacetime and $\phi\in\Sol(\Mb)$. We may write $\phi=E_\Mb f$ for some
$f\in \CoinX{\Mb;\CC^{|\nu|}}$ and then use a partition of unity to write $f$ as a finite sum $f=\sum_i f_i$ where each $f_i$ is supported in a diamond $D_i$ in $\Mb$. Then it is evident
that $\phi$ is contained in the span of $E_\Mb \CoinX{D_i;\CC^{|\nu|}}$, which is the image of $\Uf(\Sol(\iota_{\Mb;D_i}))$. Hence $\Sol$ is $\Uf$-additive. 
 
The endomorphisms of $\Sol$ will now be classified in terms of the group
\[
{\rm O}(\nu) = \prod_{m\in\Mc} {\rm O}(\nu(m)),
\]
where the factors are the standard groups of real orthogonal matrices.

\begin{Lem} \label{lem:Onu_action}
Each $R=(R_m)_{m\in\Mc}\in{\rm O}(\nu)$ induces an automorphism $S(R)$
of $\Sol$ given by
\begin{equation}\label{eq:SR_action}
S(R)_\Mb = \bigoplus_{m\in\Mc} \II_{\Sol_m(\Mb)}\otimes R_m,
\end{equation}
and the map $R\mapsto S(R)$ is an injective group homomorphism of $O(\nu)$
into $\Aut(\Sol)$.
\end{Lem}
\begin{proof} It is easily seen that $S(R)_\Mb$ is an endomorphism
of $\Sol(\Mb)$ and has inverse $S(R^{-1})_\Mb$. If
$\psi:\Mb\to\Nb$ then
\begin{align*}
S(R)_\Nb\circ \Sol(\psi)\bigoplus_{m\in\Mc} \phi_m\otimes z_m &= 
S(R)_\Mb \bigoplus_{m\in\Mc} (\Sol_m(\psi)\phi_m)\otimes z_m \\
&=
\bigoplus_{m\in\Mc} \Sol_m(\psi)\phi_m\otimes R_m z_m
= \Sol(\psi) \bigoplus_{m\in\Mc} \phi_m\otimes z_m \\
&= 
\Sol(\psi) \circ S(R)_\Mb \bigoplus_{m\in\Mc} \phi_m\otimes z_m
\end{align*}
for arbitrary $\bigoplus_{m\in\Mc} \phi_m\otimes z_m \in \Sol(\Mb)$
and (extending by linearity) we see that $S(R)$ is natural.
Thus $S(R)\in\Aut(\Sol)$; the homomorphism and injectivity properties are clear. 
\end{proof}

The main result of this section is that these are the only endomorphisms of $\Sol$.
\begin{Thm}\label{thm:classical_endos}
Every endomorphism of $\Sol$ is an automorphism, and
\[
\End(\Sol)=\Aut(\Sol)\cong {\rm O}(\nu),
\]
with the isomorphism given by the homomorphism of Lemma~\ref{lem:Onu_action}.
\end{Thm}
\begin{proof} Owing to the timeslice property, $\Uf$-additivity and
Theorem~\ref{thm:reduce_to_Mink},  any $\eta\in \End(\Sol)$ is uniquely determined by its Minkowski space component $\eta_{\Mb_0}$. Using the fact that
$\eta_{\Mb_0}$ preserves the stress-energy tensor and commutes with
translations of Minkowski space and complex conjugation, Proposition~\ref{prop:SET} in Appendix~\ref{sect:SET} shows that  there are real orthogonal matrices $R_m$ such that
$\eta_{\Mb_0} = S(R)_{\Mb_0}$, where 
$R=(R_m)_{m\in\Mc}\in {\rm O}(\nu)$. Thus $\eta = S(R)\in\Aut(\Sol)$. In summary, all endomorphisms are automorphisms, and the homomorphism of Lemma~\ref{lem:Onu_action} is surjective and hence an isomorphism. 
\end{proof}

\subsection{Quantized theory: Field algebra}

Any $(V,\sigma,C)\in \Sympl$ has  a quantization given by
the unital $*$-algebra $\Qf(V,\sigma,C)$, whose underlying complex vector space is the symmetric tensor vector space over $V$,
\begin{equation}\label{eq:Q1}
\Qf(V,\sigma,C) = \Gamma_\odot(V) \stackrel{\text{def}}{=} \bigoplus_{n\in\NN_0} V^{\odot n} ,
\end{equation}
(all tensor products and direct sums being purely algebraic) 
with a product such that
\begin{equation}\label{eq:Q2}
u^{\odot m} \cdot v^{\odot n} = \sum_{r=0}^{\min\{m,n\}}
\left(\frac{i\sigma(u,v)}{2}\right)^r \frac{m!n!}{r!(m-r)!(n-r)!}
S\left( u^{\otimes (m-r)}\otimes v^{\otimes (n-r)}\right),
\end{equation}
where $S$ denotes symmetrisation, and a $*$-operation defined by $(u^{\odot n})^* = (C u)^{\odot n}$; both operations being extended by (anti-)linearity to general elements of $\Gamma_\odot(V)$. By convention $u^{\odot 0}=1\in V^{\odot 0}=\CC$. 
The nondegeneracy of $\sigma$ ensures that the algebra $\Qf(V,\sigma,C)$ is simple. 
Moreover, the quantization becomes a functor from $\Sympl$ to $\Alg$ if we also assign
\begin{equation}\label{eq:Q3}
\Qf(f) = \Gamma_\odot(f) = \bigoplus_{n\in\NN_0}^\infty f^{\odot n}
\end{equation}
to any morphism $f:(V,\sigma,C)\to (V',\sigma',C')$ in $\Sympl$. We refer
to~\cite{FewVer:dynloc2} for a full exposition. 

The quantization of the classical theory $\Sol$ studied in the previous subsection is given by $\Ff = \Qf\circ \Sol$,\footnote{In \cite{FewVer:dynloc2},
the analogous theory was denoted as $\Af$, as is traditional in QFT in CST;
here, we adopt the AQFT convention, that field algebras are denoted with an `F' and observable algebras
with an `A'.} which is additive because $\Sol$ is $\Uf$-additive.
It will be useful to describe $\Ff$ using `symplectically smeared fields': for
each $\Mb\in\Loc$, let $\cPhi_\Mb:\Sol(\Mb)\to\Gamma_\odot(\Sol(\Mb))$ be the canonical injection $\cPhi_\Mb(\phi) = 0\oplus\phi\oplus 0 \oplus\cdots$. 
Then $\cPhi$ is a field of type $\Sol$; $\cPhi\in\Fld(\Sol,\Ff)$ [we suppress the forgetful functors from $\Sympl$ and $\Alg$ to $\Set$], obeying 
\begin{align}
\cPhi_\Mb:\phi &\to \cPhi_\Mb(\phi) \quad\text{is $\CC$-linear} \notag\\
\cPhi_\Mb(\phi)^* &=\cPhi(\overline{\phi}) \label{eq:cPhi_rels}\\
[\cPhi_\Mb(\phi),\cPhi_\Mb(\phi')] &= i\sigma_\Mb(\phi,\phi')\II_{\Ff(\Mb)}\notag
\end{align}
for all $\phi,\phi'\in \Sol(\Mb)$. Indeed, these relationships characterize the theory:
given any map $\cPsi_\Mb$ from $\Sol(\Mb)$ to some $\Ac\in\Alg$ obeying the
above relations, there is a $\Alg$-morphism\footnote{The 
morphism is necessarily monic because $\Ff(\Mb)$ is simple (and $\Ac$ is
not the zero algebra).} $\alpha:\Ff(\Mb)\to\Ac$ such that $\alpha\cPhi_\Mb(\phi)
=\cPsi_\Mb(\phi)$, 
which is an isomorphism if $\Ac$ is generated by the $\cPsi_\Mb(\phi)$ and $\II_{\Ac}$. 
The relative Cauchy evolution of $\Ff$ is closely linked to that of $\Sol$ by $\rce_\Mb^{(\Ff)}[\hb]
=\Qf(\rce_\Mb^{(\Sol)}[\hb])$ and hence
\begin{equation}\label{eq:rce_Phi}
\rce_\Mb^{(\Ff)}[\hb] \cPhi(\phi) =  \cPhi(\rce_\Mb^{(\Sol)}[\hb] \phi).
\end{equation}
The usual spacetime smeared fields of the theory are obtained by setting $\Phi_\Mb(f) = 
\cPhi_\Mb(E_\Mb f)$ for $f\in\CoinX{\Mb;\CC^{|\nu|}}\cong \bigoplus_{m\in\Mc}
\CoinX{\Mb;\CC^{\nu(m)}}$, where
\[
E_\Mb = \bigoplus_{m\in\Mc} E_{m,\Mb}\otimes\II_{\nu(m)}.
\]

We introduce a state space as follows. For each $\Mb$, let $\Sf(\Mb)$ be the set of all
states $\omega$ on $\Ff(\Mb)$ such that all $k$-point functions are distributional, 
in the sense that
\[
\CoinX{\Mb;\CC^{|\nu|}}^{\otimes k}\owns f_1\otimes \cdots\otimes f_k \mapsto 
\omega(\Phi_\Mb( f_1)\cdots \Phi_\Mb( f_k))
\]
is continuous in the usual test-function topology for each $k\in\NN$. This is a much larger state space than
the Hadamard class usually employed; our results would be unaltered by restricting
to this class. It is easily seen that $\Sf(\Mb)$ is closed under convex linear combinations and with respect to operations induced by $\Ff(\Mb)$; moreover,
the pull-back of a distribution by a smooth embedding is also a distribution, so 
$\Ff(\psi)^*\Sf(\Nb)\subset\Sf(\Mb)$ for any $\psi:\Mb\to\Nb$. Thus,
$\Sf(\psi)$ may be defined uniquely so that the diagram 
\[
\begin{tikzpicture}[description/.style={fill=white,inner sep=2pt}]
\matrix (m) [ampersand replacement=\&,matrix of math nodes, row sep=2em,
column sep=4em, text height=1.5ex, text depth=0.25ex]
{  \Sf(\Nb) \&  \Sf(\Mb) \\
  \Ff(\Nb)^* \&  \Ff(\Mb)^* \\ };
\path[->]
(m-1-1) edge[right hook->]  (m-2-1)
(m-1-1) edge node[auto] {$\Sf(\psi)$} (m-1-2)
(m-1-2) edge[right hook->]  (m-2-2)
(m-2-1) edge node[auto] {$\Ff(\psi)^*$} (m-2-2);
\end{tikzpicture}
\]
commutes, where the unlabelled maps are the obvious inclusions. With
this definition, $\Sf$ becomes a state space functor. Combining $\Ff$ and
$\Sf$ we obtain a theory 
$\langle \Ff,\id,\Sf\rangle\in\LCT_\grAS$ that obeys locality in the sense
of commutation at spacelike separation.

Our task is now to classify $\End(\langle \Ff,\id,\Sf\rangle)$,  i.e.,
natural transformations $\eta:\Ff\nto\Ff$ such that $\eta_\Mb^*\Sf(\Mb)\subset\Sf(\Mb)$ for all $\Mb$.\footnote{While it would be interesting to obtain a (purely algebraic) classification of $\End(\Ff)$, our current proof involves continuity arguments that 
require the specification of a state space.}
As a first observation, note that any $\zeta\in\End(\Sol)$ lifts 
to an endomorphism $\Qf[\zeta]$ of $\Ff$ with components
$\Qf[\zeta]_\Mb = \Qf(\zeta_\Mb)$, and that $\Qf[\zeta]$ is an automorphism
if $\zeta\in\Aut(\Sol)$, because functors preserve isomorphisms. 
Thus Theorem~\ref{thm:classical_endos}
shows that ${\rm O}(\nu)\owns R\mapsto \Qf[S(R)]\in\Aut(\Ff)$ is a group homomorphism. For theories with no massless components (i.e., $\nu(0)=0$) this
will turn out to give the full class of endomorphisms of $\langle \Ff,\id,\Sf\rangle$. However,  massless theories admit additional automorphisms
and the general result Theorem~\ref{thm:main} is that $\Aut(\langle \Ff,\id,\Sf\rangle)$ is isomorphic to
\[
G(\nu) = {\rm O}(\nu) \ltimes \RR^{\nu(0)*},
\]
where $\RR^{k*}$ is the additive group of real $k$-dimensional
row vectors (with the convention that this is the trivial group if
$k=0$) and the semidirect product is defined by
\[
\left((R_m)_{m\in\Mc},\ell\right)\cdot
\left((R'_m)_{m\in\Mc},\ell'\right)
= \left((R_mR'_m)_{m\in\Mc},\ell R_0' +\ell'\right).
\]
To explain the action of $G(\nu)$, we require some notation. 
For any $\phi\in\Sol(\Mb)$, let $\phi_0$ be the component
of $\phi$ in $\Sol_0(\Mb)\otimes\CC^{\nu(0)}$, regarded as a $\nu(0)$-dimensional column vector with entries in $\Sol_0(\Mb)$. Similarly, for $f\in\CoinX{\Mb;\CC^{|\nu|}}$, $f_0$ denotes the component in $\CoinX{\Mb;\CC^{\nu(0)}}$, so that
$\phi_0 = E_{0,\Mb}f_0$. With these conventions, we set
\begin{equation}\label{eq:ell_def}
\langle \ell, E_\Mb f\rangle_\Mb=\int_\Mb \ell \cdot f_0\,\dvol_\Mb ,
\end{equation}
which is well-defined because all elements of $\ker E_{0,\Mb}=\Box_\Mb\CoinX{\Mb}$
have vanishing spacetime integral. Equivalently, we have
$\langle \ell, \phi\rangle_\Mb=\sigma_{\Mb}(\ell\cdot\phi_0, 1_\Mb)$, where
$1_\Mb$ is the unit constant on $\Mb$ (extending the notation
if $\Mb$ has noncompact Cauchy surfaces, so $1_\Mb\neq \Sol(\Mb)$). 
It is easily 
verified that 
\begin{equation}\label{eq:ell_covariance}
\langle \ell, \Sol(\psi)\phi\rangle_\Nb = \langle \ell, \phi\rangle_\Mb
\end{equation} 
holds for all $\phi\in\Sol(\Mb)$ and $\psi:\Mb\to\Nb$.
The group  $G(\nu)$ acts in the following way. 
\begin{Prop} \label{prop:Gnu_action}
There is a  group monomorphism $\zeta:G(\nu)\to\Aut(\langle \Ff,\id,\Sf\rangle)$
such that 
\begin{equation}\label{eq:zeta_defn}
\zeta(R,\ell)_\Mb\cPhi_\Mb(\phi) = 
\cPhi_\Mb(S(R)_\Mb\phi) + \langle\ell, \phi\rangle_\Mb\II_{\Ff(\Mb)}
\end{equation}
for all $\Mb\in\Loc$, $\phi\in \Sol(\Mb)$ and $(R,\ell)\in G(\nu)$.
\end{Prop}
\begin{proof} Let $(R,\ell)\in G(\nu)$ be arbitrary, and write the
right-hand side of 
\eqref{eq:zeta_defn} as 
$\cPsi_\Mb(\phi)$. 
Then $\phi\mapsto\cPsi_\Mb(\phi)$ is a
linear map of $\Sol(\Mb)$ to $\Ff(\Mb)$ obeying $\cPsi_\Mb(\phi)^*=\cPsi_\Mb(\overline{\phi})$ and
$[\cPsi_\Mb(\phi),\cPsi_\Mb(\phi')] =
i\sigma_\Mb(\phi,\phi')\II_{\Ff(\Mb)}$.
Thus there is a unique endomorphism $\zeta(R,\ell)_\Mb$ of $\Ff(\Mb)$ such
that \eqref{eq:zeta_defn} holds. This is invertible, with inverse
$\zeta((R,\ell)^{-1})_\Mb$, so $\zeta(R,\ell)_\Mb\in\Aut(\Ff(\Mb))$. 
Moreover, if a state $\omega$ has distributional $k$-point functions, so does
$\zeta(R,\ell)_\Mb^*\omega$; here, we use the fact that $\langle \ell, E_\Mb f\rangle_\Mb$
is evidently continuous in $f\in\CoinX{\Mb;\CC^{|\nu|}}$.  Thus $\zeta(R,\ell)_\Mb^*\Sf(\Mb)= 
\Sf(\Mb)$, as $\zeta(R,\ell)_\Mb$ is an isomorphism.
Moreover, we have 
\begin{align*}
\Ff(\psi)\zeta(R,\ell)_\Mb\cPhi_\Mb(\phi) &=
\Ff(\psi)\left(\cPhi_\Mb(S(R)_\Mb\phi) +
\langle\ell,\phi\rangle_{\Mb}\II_{\Ff(\Mb)} \right)\\
&=
\cPhi_\Nb(\Sol(\psi) S(R)_\Mb\phi) +
\langle\ell,\phi\rangle_{\Mb}\II_{\Ff(\Nb)} \\
&=
\cPhi_\Nb(S(R)_\Nb \Sol(\psi) \phi) +
\langle\ell,\Sol(\psi)\phi\rangle_{\Nb}\II_{\Ff(\Nb)}  \\
&= \zeta(R,\ell)_\Nb\Ff(\psi)\cPhi_\Mb(\phi),
\end{align*}
for any $\Mb\stackrel{\psi}{\to}\Nb$, using $S(R)\in\End(\Sol)$ and
Eq.~\eqref{eq:ell_covariance}, so the components $\zeta(R,\ell)_\Mb$ form a natural transformation $\zeta(R,\ell):\Ff\nto\Ff$, whereupon
$\zeta(R,\ell)\in\Aut(\langle \Ff,\id,\Sf\rangle)$. 
That $\zeta:G(\nu)\owns(R,\ell)\mapsto \zeta(R,\ell)\in
\Aut(\langle \Ff,\id,\Sf\rangle)$ is a homomorphism
holds because
\begin{align*}
\zeta(R,\ell)_\Mb \zeta(R',\ell')_\Mb \cPhi_\Mb(\phi) &= 
\zeta(R,\ell)_\Mb \left[ \cPhi(S(R')_{\Mb}\phi) +
\langle\ell',\phi\rangle_\Mb\II_{\Ff(\Mb)}\right] \\
&=
\cPhi(S(R)_{\Mb}S(R')_{\Mb}\phi) + \left(
\langle\ell,S(R')_{\Mb}\phi\rangle+
\langle\ell',\phi\rangle_\Mb\right)\II_{\Ff(\Mb)} \\
&=\zeta(RR',\ell R'_0+\ell')\cPhi_\Mb(\phi),
\end{align*}
and the proof that $\zeta$ is a monomorphism is straightforward. \end{proof}

The main result of this section is that every endomorphism of $\langle\Ff, \id,\Sf\rangle$ 
is one of the automorphisms just constructed. Our proof uses the 
Fock representation of the Minkowski vacuum state; as there is no such state
for massless fields in $n=2$ spacetime dimensions, we must exclude this
case from the current treatment. 
\begin{Thm} \label{thm:main}
For spacetime dimension $n\ge 3$, every endomorphism of $\langle\Ff, \id,\Sf\rangle$  is an automorphism. Moreover,
the monomorphism $\zeta:G(\nu)\to \Aut(\langle\Ff, \id,\Sf\rangle)$ of Proposition~\ref{prop:Gnu_action}
is an isomorphism of groups, so we have
\[
\End(\langle\Ff, \id,\Sf\rangle) = \Aut(\langle\Ff, \id,\Sf\rangle) \cong G(\nu).
\]
The same result holds in two dimensions if $\nu(0)=0$, whereupon
 $\Aut(\langle\Ff, \id,\Sf\rangle) \cong {\rm O}(\nu)$. 
\end{Thm}
{\noindent\bf Remark:} In the purely massive case ($\nu(0)=0$), it is well-known that the maximal DHR group for the affiliated local $C^*$-algebraic net 
is precisely $G_{\rm max} = {\rm O}(\nu)$. Thus the conclusion of
Theorem~\ref{thm:end_aut} holds in this case (with $G_\rce=G_{\rm max}$). 
A direct application of the theory of Sec.~\ref{sect:compactness}
to the Weyl algebra is described in Sec.~\ref{sec:Weyl}. 

\begin{proof} Given any $\eta\in\End(\langle\Ff, \id,\Sf\rangle)$, the
transformed symplectically smeared field  $\cPsi=\eta\cdot\cPhi$ obeys many
of the same properties as $\cPhi$; namely, its behaviour under adjoints
\begin{equation}\label{eq:Psi_adjoint}
\cPsi_\Mb(\phi)^* = (\eta_\Mb\cPhi_\Mb(\phi))^* = 
\eta_\Mb (\cPhi_\Mb(\phi))^* = \eta_\Mb\cPhi_\Mb(\overline{\phi}) =
\cPsi_\Mb(\overline{\phi})
\end{equation}
and the commutation relations
\begin{equation}\label{eq:Psi_CCRs}
[\cPsi_\Mb(\phi),\cPsi_\Mb(\phi')] = \eta_\Mb[\cPhi_\Mb(\phi),\cPhi_\Mb(\phi')]
= i\sigma_\Mb(\phi,\phi')\II_{\Ff(\Mb)}.
\end{equation}
Moreover, Eqs.~\eqref{eq:rce_Phi} and \eqref{eq:intertwine} entail that
\begin{equation}\label{eq:Psi_rce}
\rce_\Mb^{(\Ff)}[\hb] \cPsi_\Mb(\phi) = \cPsi_\Mb(\rce_\Mb^{(\Sol)}[\hb] \phi).
\end{equation}

We now focus on Minkowski space $\Mb_0$. Using properties of
the vacuum representation, we show in Proposition~\ref{prop:Minkowski} below
that $\cPsi_{\Mb_0}$ takes the form
\begin{equation} 
\cPsi_{\Mb_0}(\phi) = 
\cPhi_{\Mb_0}(S\phi) + \langle\ell, \phi\rangle_{\Mb_0}\II_{\Ff(\Mb_0)}
\qquad (\phi\in\Sol(\Mb_0))
\end{equation}
for some linear map $S:\Sol(\Mb_0)\to\Sol(\Mb_0)$ and $\ell\in\CC^{\nu(0)*}$ [with $\ell=0$ if $\nu(0)=0$]. Substituting in Eq.~\eqref{eq:Psi_rce}
and using Eq.~\eqref{eq:rce_Phi}, we find
\[
\cPhi_{\Mb_0}(\rce_{\Mb_0}^{(\Sol)}[\hb] S\phi) +
\langle\ell, \phi\rangle_{\Mb_0}\II_{\Ff(\Mb_0)} = 
\cPhi_{\Mb_0}(S \rce_{\Mb_0}^{(\Sol)}[\hb]\phi) + 
\langle\ell, \rce_{\Mb_0}^{(\Sol)}[\hb]\phi\rangle_{\Mb_0}\II_{\Ff(\Mb_0)}
\]
and may deduce that $S$ commutes with the relative Cauchy evolution in $\Mb_0$
(the identity $\langle\ell, \rce_{\Mb_0}^{(\Sol)}[\hb]\phi\rangle_{\Mb_0}=\langle\ell, \phi\rangle_{\Mb_0}$
gives no additional constraint). 
Accordingly, $\phi$ and $S\phi$ have identical (classical)
stress-energy tensors; moreover $S$ commutes with the action of spacetime translations by naturality of $\eta$, and obeys $S\overline{\phi}=\overline{S\phi}$
owing to Eq.~\eqref{eq:Psi_adjoint}.  Proposition~\ref{prop:SET} implies that 
$S=S(R)_{\Mb_0}$ for some $R=(R_m)_{m\in\Mc}\in{\rm O}(\nu)$, and
Eq.~\eqref{eq:Psi_adjoint} also shows that $\ell$ is real.  Hence $\eta_{\Mb_0} = \zeta_{\Mb_0}(R,\ell)$, and as $\langle\Ff,\id,\Sf\rangle$
obeys the timeslice property and is $\Uf$-additive with respect to the
forgetful functor $\Uf:\grAS\to\Alg$, 
Theorem~\ref{thm:reduce_to_Mink} entails that $\eta=\zeta(R,\ell)\in
\Aut(\langle\Ff,\id,\Sf\rangle)$, and the result  follows immediately. \end{proof}

The remaining task is to prove Proposition~\ref{prop:Minkowski}. Our argument
uses particular features of the standard Minkowski vacuum state $\omega$, which
exists for $n\ge 3$, or $n=2$ provided $\nu(0)=0$. Let us briefly recall
some properties of the induced GNS representation $(\Fc,\pi,\Dc,\Omega)$
of $\Ff(\Mb_0)$. First, we will denote
the one-particle Hilbert space by $\Hc$, so $\Fc=\Fc_{\odot}(\Hc)$.
Second, to each translation $\tau^{y}(x)=x+y$,
there is a unitary $U(y)=\exp(iy^a P_a)$ so that $U(y)\pi(A)U(y)^{-1} = \pi(\Ff(\tau^y)A)$
for all $A\in\Ff(\Mb_0)$. The momentum operators $P_a$ are defined on a
domain in $\Fc$
including $\Dc=\pi(\Ff(\Mb_0))\Omega$, on which $y\mapsto U(y)$ is
therefore strongly differentiable; the joint spectrum of the $P_a$ lies
in the closed forward lightcone. Third, the discrete spectrum of the 
mass-squared operator $P^a
P_a$ is $\sigma_{\rm disc}(P^a P_a) = \{0\}\cup\{m^2:m\in\Mc\}$. 
Denoting the eigenspace of eigenvalue
$m^2$ by $\Hc_m$, the subspaces $\Hc_m$ for $m>0$ are subspaces of the
one-particle space $\Hc$, while $\Hc_0$ is spanned by the one-particle
states of any massless fields and the vacuum vector, i.e., 
\[
\Hc_0 = \CC\Omega \oplus \left(\bigoplus_{m\in\Mc\backslash\{0\}} \Hc_m
\right)^\perp \subset \Fc,
\]
where the orthogonal complement is taken in $\Hc$. Fourth, 
the vacuum vector $\Omega$ is separating for the representation $\pi$, i.e., $\pi(A)\Omega=0$ 
iff $A=0$.  Fifth, we have
the following:
\begin{Lem} \label{lem:lift}
Suppose $A\in\Ff(\Mb_0)$ is such that
$\pi(A)\Omega\in\Hc_m$ for some $m\ge 0$. Then
$A = \cPhi_{\Mb_0}(\phi) + \alpha\II_{\Ff(\Mb_0)}$
for unique $\phi\in\Sol(\Mb_0)$ and $\alpha\in\CC$, which obey
$\Mc\phi=m\phi$. If $m>0$, $\alpha$ must vanish.
\end{Lem}
\begin{proof} As $\Ff(\Mb_0)=\Gamma_\odot(\Sol(\Mb_0))$,
each nonzero $A\in \Ff(\Mb_0)$ therefore has a degree $\deg A$, which is the maximum $n\in\NN_0$
such that $A^{(n)}\neq 0$, where $A^{(n)}$ is the component of $A$ in $\Sol(\Mb_0)^{\odot n}$;
we assign a degree $-1$ to the zero element of $\Ff(\Mb_0)$. 
Moreover, the vacuum $\Omega$ induces a normal ordering operation on $\Ff(\Mb_0)$, so that
$\pi({:}A{:}_{\Omega})={:}\pi(A){:}$, where the normal ordering on the right-hand side is that
of the Fock space $\Fc$. A key fact is that $\deg(A-{:}A{:}_\Omega) = \max\{-1,\deg(A)-2\}$. 
Now $\deg A$ is also the maximum $n\in\NN_0$ for which
$\Upsilon=\pi(A)\Omega$ has nonvanishing component in the $n$-particle space in Fock space:
to see this, note that $\Upsilon$ clearly has vanishing component in all higher $n$-particle spaces,
while its projection onto the $\deg A$-particle space coincides with ${:}\pi(A){:}\Omega = 
\pi({:}A{:}_\Omega)\Omega$; if this should vanish then ${:}A{:}_\Omega=0$ because $\Omega$ is separating for $\pi$, and hence $\deg A=\max\{\deg(A)-2,-1\}$, which implies $A=0$. 

Now suppose that $\pi(A)\Omega\in\Hc_m$, which lies in the $1$-particle space if $m\neq 0$, 
or the span of the $0$ and $1$-particle spaces if $m=0$. Accordingly, if $A\neq 0$,
we must have $\deg A=1$ for $m>0$ and $\deg A\le 1$ if $m=0$. The result follows.  \end{proof}

Finally, we may state and prove the remaining result. 
\begin{Prop} \label{prop:Minkowski}
Let $\eta\in\End(\langle\Ff, \id,\Sf\rangle)$ and define $\cPsi = \eta\cdot\cPhi$. Then
there exists a unique linear map $S:\Sol(\Mb_0)\to\Sol(\Mb_0)$ and
$\ell\in\CC^{\nu(0)*}$ [with $\ell=0$ if $\nu(0)=0$] and
\[
\cPsi_{\Mb_0}(\phi) = \cPhi_{\Mb_0}(S\phi) +
\langle\ell,\phi\rangle_{\Mb_0}\II_{\Ff(\Mb_0)},
\]
for all $\phi\in \Sol(\Mb_0)$.
\end{Prop}
\begin{proof} In the vacuum representation,
consider vectors of the form $\pi(\cPsi_{\Mb_0}(\phi))\Omega$ for
$\phi\in \Sol(\Mb_0)$. The properties of $U(y)$ mentioned above entail that
\begin{equation}\label{eq:UPsi}
U(y)\pi(\cPsi_{\Mb_0}(\phi))\Omega = \pi(\Ff(\tau^y)\cPsi_{\Mb_0}(\phi))\Omega
= \pi(\cPsi_{\Mb_0}(L(\tau^y)\phi))\Omega .
\end{equation}
Let $e_a$ ($a=0,\ldots,n-1$) be standard inertial basis vectors. Putting
$y=s e_a$, the left-hand side may be differentiated with respect to $s$ to
give $i P_a \pi(\cPsi_{\Mb_0}(\phi))\Omega$ at $s=0$. 
To evaluate the derivative of the right-hand side, we set $\phi=E_\Mb f$
for $f\in\CoinX{\Mb_0;\CC^{|\nu|}}$. Then
\[
s^{-1}(\Sol(\tau^{s e_a})\phi - \phi) +  \nabla_a \phi = 
E_{\Mb_0} \left( s^{-1}[\tau^{s e_a}_*f - f] +  \nabla_a f\right)  
\]
and the parenthesis on the right-hand side tends to $0$ in $\CoinX{\Mb_0;\CC^{|\nu|}}$. 
Now 
\[
\|\pi(\cPsi(E_{\Mb_0}h))\Omega\|^2 = (
\eta_{\Mb_0}^*\omega)(\cPhi(E_{\Mb_0}\overline{h})\cPhi(E_{\Mb_0}h )),
\]
which is a $2$-point function for the state $\eta_{\Mb_0}^*\omega\in\Sf(\Mb_0)$,
and therefore a distribution.  Thus 
$\pi(\cPsi(E_{\Mb_0}h))\Omega\to 0$
in $\Fc$ as $h\to 0$ in $\CoinX{\Mb_0;\CC^{|\nu|}}$ and the left-hand side of~\eqref{eq:UPsi} has
the derivative one would expect, namely $-\pi(\cPsi_{\Mb_0}(\nabla_a\phi))\Omega$. Accordingly, $P_a \pi(\cPsi_{\Mb_0}(\phi))\Omega =
\pi(\cPsi_{\Mb_0}(i\nabla_a\phi))\Omega$
for any $\phi\in\Sol(\Mb_0)$ and hence
\begin{equation}\label{eq:Psquared}
P_a P^a \pi(\cPsi_{\Mb_0}(\phi))\Omega =
\pi(\cPsi_{\Mb_0}(-\Box_{\Mb_0}\phi))\Omega = \pi(\cPsi_{\Mb_0}(\Mc^2\phi))\Omega ,
\end{equation}
where we have abused notation by writing
$\Box_\Mb$ and $\Mc^2$ as a shorthand for the operators
\[
\bigoplus_{m\in\Mc} \Box_\Mb\otimes\II_{\nu(m)}, \qquad
\bigoplus_{m\in\Mc} m^2 \II_{L(\Mb)}\otimes \II_{\nu(m)}
\]
on $\Sol(\Mb)$. 
From Eq.~\eqref{eq:Psquared} we see that  $\Mc\phi=m\phi$ implies $\pi(\cPsi_{\Mb_0}(\phi))\Omega\in \Hc_m$. According to
Lemma~\ref{lem:lift}, we may therefore deduce that, for general $\phi\in \Sol(\Mb_0)$,
\begin{equation}\label{eq:cPsi_intermediate}
\cPsi_{\Mb_0}(\phi) = \cPhi_{\Mb_0}(S\phi) + \alpha(\phi_0)\II_{\Ff(\Mb_0)},
\end{equation}
where $\phi_0$ is the component of $\phi$ in
$\Sol_0(\Mb_0)\otimes\CC^{\nu(0)}$ and $S:\Sol(\Mb_0)\to \Sol(\Mb_0)$ and
$\alpha:\Sol_0(\Mb_0)\otimes\CC^{\nu(0)}\to\CC$ are uniquely determined and necessarily linear maps
(we can also deduce $S\Mc=\Mc S$).  
It remains to determine $\alpha$. 
By Eq.~\eqref{eq:cPsi_intermediate}, it is clear that
\[
\alpha(\phi_0) = \ip{\Omega}{\cPsi_{\Mb_0}(\phi)\Omega} =
(\eta_{\Mb_0}^*\omega)(\cPhi_{\Mb_0}(\phi)).
\]
As $\eta_{\Mb_0}^*\omega$ is a translationally invariant
state in $\Sf(\Mb_0)$ by Lemma~\ref{lem:invariant_states}, it follows
that 
\[
\DD(\Mb_0)\otimes\CC^{\nu(0)} \owns f\mapsto  
(\alpha\circ( E_{0\,\Mb_0}\otimes \II_{\nu(0)}))(f) \in \CC
\]
may be regarded as a row vector of translationally invariant distributions,
each component of which must be constant.
Thus there is $\ell\in\CC^{\nu(0)*}$ such that
\[
(\alpha\circ (E_{0,\Mb_0}\otimes\II_{\nu(0)}))(f) = 
\int_{\Mb_0} \ell\cdot f \dvol_{\Mb_0}
\] 
and in fact we must have $\ell\in\RR^{\nu(0)*}$ to obtain 
$\alpha(\overline{\phi_0}) = \overline{\alpha(\phi_0)}$,
whereupon 
$\alpha(\phi_0) =\langle \ell,\phi\rangle_{\Mb_0}$
for all $\phi\in \Sol(\Mb_0)$, completing the proof.  
\end{proof}

\subsection{Quantized theory: Algebra of Observables}
\label{sect:observables}

In each spacetime $\Mb$, the algebra of observables $\Af(\Mb)$ may be concretely
constructed as the subalgebra of the field algebra $\Ff(\Mb)$ of elements invariant under $\zeta(R,\ell)_\Mb$ for all $(R,\ell)\in G(\nu)={\rm O}(\nu)\ltimes \RR^{\nu(0)*}$. 

We begin with the issue of $\RR^{\nu(0)*}$-invariance (assuming $\nu(0)>0$).  Let $\widehat{\Sol}_0(\Mb)$ be the subspace of $\Sol_0(\Mb)$ consisting of solutions with vanishing symplectic product with the constant unit solution. We call these `charge-zero' solutions, because this symplectic product
is precisely the Noether charge corresponding to the rigid gauge freedom to add a constant solution
(the same constant in all connected components of $\Mb$). 
The subspace $\widehat{\Sol}_0(\Mb)$ has codimension $1$ in $\Sol_0(\Mb)$; we choose any $\theta\in\Sol_0(\Mb)$ with  $\sigma_{0,\Mb}(\theta,1_\Mb)=1$, which then spans a complementary subspace to $\widehat{\Sol}_0(\Mb)$ in $\Sol_0(\Mb)$. For notational simplicity,
it is also convenient to write $\widehat{\Sol}_m(\Mb)=\Sol_m(\Mb)$ for any $m>0$. 
With these choices, any 
$A\in\Af(\Mb)$ may be written in the form
\[
A = \sum_{k=0}^{\deg A} \sum_{|\alpha|\le k} 
S \left(\left(\bigotimes_{i=1}^{\nu(0)} (\theta\otimes e_i)^{\otimes \alpha_i}\right)
\otimes Z_{k,\alpha}\right),
\]
where the sum runs over all muti-indices $\alpha$ of total order $|\alpha|\le k$,
the $e_i$ are a standard real basis for $\CC^{\nu(0)}$, $S$ is the symmetrisation operator
and 
\[
Z_{k,\alpha}\in \widehat{\Sol}(\Mb)^{\odot (k-|\alpha|)}, \quad\text{where}\quad 
\widehat{\Sol}(\Mb):=
\bigoplus_{m\in \Mc} \widehat{\Sol}_m(\Mb)\otimes\CC^{\nu(m)}.
\]
Let $e_i^*\in\RR^{\nu(0)*}$ be the dual basis to $e_i$. Then because
the polynomial $\lambda\mapsto \zeta(\II,\lambda e_j^*)_\Mb A$ is constant, the coefficient of  
$\lambda$ must vanish, i.e., 
\[
\sum_{k=1}^{\deg A} \sum_{|\alpha|\le k} \alpha_j
S \left(\left(\bigotimes_{i=1}^{\nu(0)} (\theta\otimes e_i)^{\otimes (\alpha_i-\delta_{ij})}\right)
\otimes Z_{k,\alpha}\right) = 0,
\] 
whereupon every $Z_{k,\alpha}$ vanishes for which $\alpha_j>0$. (To see this, one works
downwards in degree.) As $1\le j\le \nu(0)$ is arbitrary, this
gives $Z_{k,\alpha}=0$ for all $|\alpha|>0$. Accordingly, all nontrivial generators of $A\in\Af(\Mb)$
belong to $\widehat{\Sol}(\Mb)$, so $\Af(\Mb)\subset \Gamma_{\odot}(\widehat{\Sol}(\Mb))$. 

Turning to the ${\rm O}(\nu)$ invariance, let us now suppose that $A\in \Gamma_{\odot}(\widehat{\Sol}(\Mb))$ obeys $\zeta(R,0)_\Mb A = A$ for all $R\in {\rm O}(\nu)$. Because $\zeta(R,0)_\Mb = \Gamma_\odot(S(R)_\Mb)$, where $S(R)_\Mb$ is defined in 
Eq.~\eqref{eq:SR_action}, the component $A_k$ of $A$ in each $\widehat{\Sol}(\Mb)^{\otimes k}$
must be invariant under $S(R)_\Mb^{\otimes k}$. Now we may identify
\begin{align*}
\widehat{\Sol}(\Mb)^{\otimes k} &
=\left(\bigoplus_{m\in\Mc}\widehat{\Sol}_m\otimes\CC^{\nu(m)}\right)^{\otimes k} \cong 
\bigoplus_{\underline{m}\in \Mc^{\times k}}
\bigotimes_{i=1}^k (\widehat{\Sol}_{m_i}(\Mb)\otimes \CC^{\nu(m_i)}) \\
&\cong \bigoplus_{\underline{m}\in \Mc^{\times k}}
\bigotimes_{m'\in\underline{m}} 
\left(\left(\widehat{\Sol}_{m'}(\Mb)^{\otimes \mu_{\underline{m}}(m')}\right) \otimes (\CC^{\nu(m')})^{\otimes \mu_{\underline{m}}(m')}\right),
\end{align*}
where $\underline{m}\in \Mc^{\times k}$ is a $k$-tuple $\underline{m}=\langle m_1,\ldots,m_k\rangle$, 
and $\mu_{\underline{m}}(m')$ is the multiplicity of $m'$ as an element of $\underline{m}$, 
and the product in the last expression is indexed over elements of $\underline{m}$ {\em disregarding
multiplicity}. With respect to the last decomposition, we have
\[
 S(R)_\Mb^{\otimes k} = \bigoplus_{\underline{m}\in \Mc^{\times k}}
\bigotimes_{m'\in\underline{m}} \id \otimes R_{m'}^{\otimes \mu_{\underline{m}}(m')}.
\]
Owing to the direct sum structure, the element $A_k$ decomposes
into components $A_{k,\underline{m}}$, each of which
is an eigenvector of unit eigenvalue for every $\bigotimes_{m'\in\underline{m}} (\id\otimes R_{m'}^{\otimes \mu_{\underline{m}}(m')})$ ($(R_{m'})_{m'\in\underline{m}}\in\prod_{m'\in\underline{m}} {\rm O}(\nu(m'))$). Some multi-linear
algebra (cf.\ e.g.,~\cite[Appendix A]{FewVer:dynloc2}) entails that 
\[
A_{k,\underline{m}}\in \bigotimes_{m'\in\underline{m}} 
\widehat{\Sol}_{m'}(\Mb)^{\otimes  \mu_{\underline{m}}(m')}
\otimes Y_{m'} ,
\]
where each $Y_{m'}\subset (\CC^{\nu(m')})^{\otimes \mu_{\underline{m}}(m')}$ is
an eigenspace of unit eigenvalue for $R^{\otimes \mu_{\underline{m}}(m')}$ for all $R\in {\rm O}(\nu(m'))$. Its elements are thus isotropic tensors (under the full orthogonal group) of rank $\mu_{\underline{m}}(m')$ in $\nu(m')$ dimensions. By classical results \cite[\S 2.9]{Weyl}, these are known to be scalars at rank $0$, products of Kronecker deltas for other even ranks, and vanishing for odd rank. 
As $A$ is an element of the symmetric
tensor vector space, we have shown that 
$\Af(\Mb)$ is contained in the $*$-subalgebra of $\Ff(\Mb)$ generated by all bilinear elements of the form
\begin{equation}\label{eq:bilinears}
\sum_{i=1}^{\nu(m)} \cPhi(\phi\otimes e_i)\cPhi(\phi'\otimes e_i)
\end{equation}
for $m\in\Mc$, $\phi,\phi'\in\widehat{\Sol}_m(\Mb)$ and with $e_i$ the standard real basis of $\CC^{\nu(m)}$, 
orthonormal with respect to the standard inner product. As this subalgebra is manifestly
invariant under the action of $G(\nu)$, we have proved:
\begin{Thm} 
The algebra of observables $\Af(\Mb)$ is the $*$-subalgebra of $\Ff(\Mb)$
generated by all bilinear elements of the form Eq.~\eqref{eq:bilinears} where
$\phi,\phi'\in\widehat{\Sol}_m(\Mb)$ and $m\in\Mc$.
\end{Thm}

In the case where $\Mb$ has more than one connected component, we see that 
the solutions $\phi,\phi'$ appearing in Eq.~\eqref{eq:bilinears} may 
have support in more than one component of the spacetime. One might not
wish to regard these as observables. Adopting the more restricted  `true algebra of observables' described at the end
of Sec.~\ref{sect:alg_of_obs}, the generating set would be restricted to 
 bilinears for which $\phi,\phi'$ have support in a single common component
of $\Mb$. 

Finally, let us specialize to the case of a single massless field,\footnote{Analogous
comments apply to any of the theories with $\nu(0)>0$.} and compare the algebra of observables obtained here with the discussion of the massless current in~\cite{FewVer:dynloc2}, in which  the classical invariance of the action under addition of locally constant solutions (which may take distinct values on different connected components of spacetime) is treated as a classical gauge symmetry. This gives a classical phase
space of gauge equivalence classes $[\phi]$ of solutions to the massless Klein--Gordon equation,\footnote{These solutions are permitted to have noncompact support,
but must be locally constant outside the causal future and past of some compact
set.} whose symplectic products with locally constant functions must be taken to vanish in order to obtain a well-defined (and, in fact, weakly non-degenerate) symplectic product on the quotient.
Quantizing, a locally covariant theory $\Cf$ is obtained, in which all $\Cf(\Mb)$ are simple, with generators $\check{\Js}_\Mb([\phi])$ obeying relations analogous to those in~\eqref{eq:cPhi_rels}. 

This is most directly comparable to the 
fixed-point subalgebra $\widehat{\FF}(\Mb)$ of $\Ff(\Mb)$ under the noncompact factor in the gauge group $\ZZ^2\ltimes\RR^*$. Here, the generators are labelled by solutions with vanishing
symplectic product with the constant solution $1_\Mb$, i.e., the space $\widehat{\Sol}_0(\Mb)$. 
If $\Mb$ is connected, and has noncompact Cauchy surfaces, then 
$\cPhi_\Mb(\phi)\mapsto \check{\Js}_\Mb([\phi])$ determines an isomorphism of the two algebras. 
However, this isomorphism breaks down if $\Mb$ is disconnected (because not all solutions in $\widehat{\Sol}_0(\Mb)$ have vanishing symplectic product with every locally constant function)
or has compact Cauchy surface (because then $\widehat{\Sol}_0(\Mb)$ contain nonzero locally constant
functions $\chi$, which correspond to nonzero elements in the centre of $\widehat{\FF}(\Mb)$; on the other hand, $[\chi]$ and hence $\check{\Js}_\Mb([\chi])$ vanish). 

The first problem may be removed by passing to the `true algebra of observables',
namely the $*$-subalgebra $\widehat{\Ff}_{o}(\Mb)$ of $\widehat{\Ff}(\Mb)$ generated by the $\check{\Js}(\phi)$ with
vanishing symplectic product with the characteristic function of each component of $\Mb$, and hence with every locally constant function. However, the second problem remains, whenever $\Mb$ has a component with compact Cauchy surface; in general $\Cf(\Mb)$ (with the modification
mentioned) is isomorphic to the quotient of $\widehat{\Ff}_o(\Mb)$ by the ideal generated by its centre. We mention that these central elements are also responsible for the failure of the theory $\Ff$
to be dynamically local~\cite{FewVer:dynloc_theory,FewVer:dynloc2}. 

\subsection{Quantized theory: Weyl algebra}\label{sec:Weyl}

We briefly explain how the symmetries of the Weyl algebra quantization may be
classified. Here, one works with the category of real symplectic spaces
$\Sympl_\RR$ and the real-valued solutions to our system, which form
a functor $\Sol_\RR:\Loc\to\Sympl_\RR$; composing with the CCR functor
$\CCR:\Sympl_\RR\to\CAlg$ gives the theory $\Wf=\CCR\circ\Sol_\RR$ -- details can be found in e.g., \cite{BrFrVe03,BarGinouxPfaffle, FewVer:dynloc2}. For simplicity, we restrict to the case in which massless fields are absent, $\nu(0)=0$; we also write $W_\Mb(\phi)$ for the
Weyl generator of $\Wf(\Mb)$ labelled by $\phi\in\Sol(\Mb)$. 
Endowing $\Wf$ with a state space $\Sf$, consisting of the
closure of the set of quasifree Hadamard states with respect to operations and local quasiequivalence,  $\langle \Wf, \id,\Sf\rangle\in\LCT_\grCAS$ obeys
assumptions (1)--(\ref{it:twduality}) of Sec.~\ref{sect:compactness} (see, e.g., 
\cite{BrFrVe03,BucWic:1986,Araki:1964} -- duality holds (at least)
for the set of double cone regions [i.e., nonempty sets of the form $I^+(p)\cap I^-(q)$ for timelike separated $p$ and $q$], which suffices 
for assumption~(\ref{it:twduality})). Hence by 
Theorem~\ref{thm:rho}, there is a faithful embedding of $\End(\langle \Wf,\id,
\Sf\rangle)$ in $G_{\rm max}$. The latter is well-known for this theory: $G_{\rm max}$ is the group of unitaries $U(R)$ ($R\in {\rm O}(\nu)$) such that $U(R)W_{\Mb_0}(\phi)U(R)^{-1} = W_{\Mb_0}(R\phi)$ and $U(R)\Omega=\Omega$. As each $R\in{\rm O}(\nu)$ induces
an automorphism of  $\langle \Wf, \id,\Sf\rangle$ by 
$\zeta(R)_\Mb W_\Mb(\phi) = W_\Mb(R\phi)$, 
assumption~(\ref{it:final}) holds with the replacement of $G_\rce$ by $G_{\rm max}$. The following is immediate:
\begin{Thm} In spacetime dimension $n\ge 3$, and subject to $\nu(0)=0$, the Weyl algebra theory $\langle \Wf, \id,\Sf\rangle$ obeys
$
\End(\langle \Wf,\id,
\Sf\rangle)=\Aut(\langle \Wf,\id,
\Sf\rangle)\cong  {\rm O}(\nu)$.
\end{Thm}
Although we have circumvented assumption~(\ref{it:rce}), the relative Cauchy evolution is unitarily implemented as a consequence of Wald's work on the $S$-matrix~\cite{Wald_Smatrix:1979}\footnote{Note,
however, that the $S$-matrix is a unitary map between two representations of the Weyl algebra for the perturbed spacetime, while the relative Cauchy evolution is an automorphism of the algebra of the unperturbed spacetime. The implementation of the relative Cauchy
evolution is (up to unitaries) the inverse of the $S$-matrix.}
and the implementors have nonzero vacuum expectation value. Thus assumption (\ref{it:rce}) also holds and we see that $G_\rce=G_{\rm max}$. We expect that the
case $\nu(0)>0$ can also be treated, resulting in the automorphism group $G(\nu)$ as in Theorem~\ref{thm:main}, but at the expense of more technicality.

\section{Conclusion}

We have argued that the global gauge group of a locally covariant
quantum field theory can be identified with the automorphism group
 of its defining functor (or, sometimes, a subgroup thereof -- see footnote~\ref{fn:subgroup}), and that this interpretation provides a natural
generalization of the standard concepts in Minkowski space AQFT. 
Furthermore, we have argued that proper endomorphisms of a 
theory are pathological, and shown that they can be excluded under
reasonable general assumptions, which also entail that the gauge group is compact. As mentioned, it is expected that this
viewpoint can contribute to the theory of superselection in curved spacetime (see, e.g., remarks in the conclusion of~\cite{Vasselli:2012arXiv1211.1812V}).
Moreover, it should have applications in other areas where locally
covariance plays a role, for example the Batalin--Vilkovisky formalism
developed by Fredenhagen and Rejzner~\cite{FreRej_BVclass:2012,FreRej_BVqft:2012} and the general discussion of
classical theories in~\cite{BruFreLRib:2012arXiv1209.2148B}. Elsewhere~\cite{FewLang:twisted} it will
also be used to give a new perspective on twisted quantum fields~\cite{Isham_twisted1:1978}.

Finally, it may be worth commenting on the special role given to Minkowski space in some of our arguments. 
From the perspective of traditional QFT in curved spacetime, it might seem 
unsatisfactory that some of our technical hypotheses are made only on the
theory as it is formulated in Minkowski space. Indeed, it is presently 
unknown how to formulate energy compactness in arbitrary spacetimes
in an elegant way. However, the assumption of energy compactness
in Minkowski space does place constraints on the theory in any
other given spacetime -- what is really lacking is a convenient technical expression for them --
and the physical interpretation is the same as in~\cite{HaaSwi:1965,BucWic:1986},
namely to restrict the number of degrees of freedom available in each given volume
of the phase space of the theory. From this perspective, our use of Minkowski space is a matter of technical simplicity. It remains an open and important issue to formulate
energy compactness in curved spacetimes, or, in the spirit of our other results,
directly at the level of the functorial description of the theory.

\smallskip

{\noindent\bf Acknowledgments} It is a pleasure to thank Henning Bostelmann, Detlev Buchholz, Roberto Conti, Klaus Fredenhagen,
Valter Moretti, Katarzyna Rejzner,  Matthew Ferguson and 
Rainer Verch for useful discussions concerning various aspects of this work.

\appendix

\section{Maps preserving the stress-energy tensor}\label{sect:SET}
\renewcommand\theThm{\Alph{section}.\arabic{Thm}}

\begin{Prop} \label{prop:SET}
Suppose $S:\Sol(\Mb_0)\to \Sol(\Mb_0)$ is a linear map [no continuity is assumed] so that $S\phi$ and $\phi$ have identical total stress-energy tensors for
all $\phi\in \Sol(\Mb_0)$ and such that $S$ commutes with
complex conjugation and the action of spacetime translations on $\Sol(\Mb_0)$. Then  there are orthogonal matrices $R_m\in {\rm O}(\nu(m))$ such that
\[
S = \bigoplus_{m\in\Mc} \II_{\Sol_m(\Mb_0)}\otimes R_m.
\]
\end{Prop}
\begin{proof} We may regard any $\phi\in \Sol(\Mb_0)$ as a smooth function
on $\Mb_0$ taking values in $\CC^{|\nu|}:=\oplus_{m\in\Mc} \CC^{\nu(m)}$. Putting the standard
norm on the latter space, we have $T_{ab}\ell^a\ell^b|_p = \| \ell^a\nabla_a \phi|_p\|^2$ for every point $p\in \Mb_0$,  null vector $\ell^a$ and  $\phi\in\Sol(\Mb_0)$; 
it follows that $S$ must obey $\| (\ell^a\nabla_a S\phi)|_p\|^2 =  \| \ell^a\nabla_a \phi|_p\|^2$. 
It is also convenient to work in terms of Cauchy data on, e.g., the 
$t=0$ surface in standard inertial coordinates on $\Mb_0$ in which the 
interval is $d\tau^2 = dt^2-\delta_{ij} dx^i dx^j$. Then
each $\phi\in\Sol(\Mb_0)$ is uniquely associated with a pair
$(\varphi,\pi) \in \CoinX{\RR^{n-1};\CC^\nu}\oplus \CoinX{\RR^{n-1};\CC^\nu}$,
where $\varphi(x) = \phi(0,x)$, $\pi(x) = (\partial\varphi/\partial t)(0,x)$ and
$S$ induces a linear map $\tS$ on $\CoinX{\RR^{n-1};\CC^{|\nu|}}\oplus \CoinX{\RR^{n-1};\CC^{|\nu|}}$
such that $(\varphi',\pi') = \tS(\varphi,\pi)$ is the Cauchy data of $S\phi$. Then
the identity $\| (\ell^a\nabla_a S\phi)|_p\|=  \| \ell^a\nabla_a \phi|_p\|$ implies
\begin{equation}\label{eq:preserve}
\left\| \left(\begin{array}{cc} \elb\cdot\nabla & \II_\nu\end{array}\right)
\tS\left(\begin{array}{c} \varphi\\ \pi\end{array}\right)(x)\right\| = 
\left\| \left(\begin{array}{cc} \elb\cdot\nabla & \II_\nu\end{array}\right)
\left(\begin{array}{c} \varphi\\ \pi\end{array}\right)(x)\right\|,
\end{equation}
for all unit vectors $\elb\in\RR^{n-1}$. In particular,  defining the map
\[
U: \left(\begin{array}{c} \varphi\\ \pi\end{array}\right) \mapsto 
\left(\begin{array}{cc} \elb\cdot\nabla & \II_{|\nu|}\end{array}\right)
\tS\left(\begin{array}{c} \varphi\\ \pi\end{array}\right)(0) \in\CC^{|\nu|}
\]
on $\CoinX{\RR^{n-1};\CC^{|\nu|}}\oplus \CoinX{\RR^{n-1};\CC^{|\nu|}}$,
we have the estimate
\begin{equation}
\label{eq:est}
\left\| U\left(\begin{array}{c} \varphi\\ \pi\end{array}\right) \right\|
\le \|(\elb\cdot\nabla\varphi)(0)\| + \|\pi(0)\|,
\end{equation}
which proves that $U$ is a $(|\nu| \times 2|\nu|)$-matrix of
distributions each of which is supported at the origin. We may therefore
conclude that
\[
U\left(\begin{array}{c} \varphi\\ \pi\end{array}\right) = A \varphi(0) +
B \pi(0) + C^j (\nabla_j\varphi)(0)  + D^j (\nabla_j\pi)(0)
\]
for $(|\nu|\times|\nu|)$-matrices $A$, $B$, $C^j$ and $D^j$, where $1\le
j\le n-1$. It is easy to show that $A=D^j=0$ for all $j$, on considering
the estimate \eqref{eq:est} in cases where $\pi=
\elb\cdot\nabla\varphi=0$, and likewise that only derivatives of
$\varphi$ along $\elb$ can contribute. This gives
\[
U\left(\begin{array}{c} \varphi\\ \pi\end{array}\right) =
R (\elb\cdot\nabla\varphi)(0) + B\pi(0)
\]
for matrices $R$ and $B$, which in principle may depend on $\elb$. Moreover, Eq.~\eqref{eq:preserve}
entails 
\[
\| R(\elb\cdot\nabla\varphi)(0) + B\pi(0)\| = \|(\elb\cdot\nabla\varphi)(0) +
\pi(0)\| ,
\]
from which we may conclude that $R$ and $B$ are unitary and in fact equal (again,
by considering cases in which $\pi(0)=0$ or $\elb\cdot\nabla\varphi(0)=0$). 
In terms of the Cauchy data $(\varphi',\pi')$ of $S\phi$,  
the discussion so far has shown that 
\[
(\elb\cdot\nabla \varphi')(0)  + \pi'(0) = R_\elb (\elb\cdot\nabla\varphi)(0) +
R_\elb\pi(0)
\]
for some unitary matrix $R_\elb$, from which we may deduce 
\begin{align}
(\elb\cdot\nabla\varphi')(0) &= \frac{1}{2}\left( 
[R_\elb + R_{-\elb}](\elb\cdot\nabla \varphi)(0) + [R_\elb - R_{-\elb}] \pi(0)
\right) \label{eq:eq1} \\
\pi'(0) &= \frac{1}{2}\left([R_\elb - R_{-\elb}] (\elb\cdot\nabla \varphi)(0)
+ [R_\elb + R_{-\elb}]\pi(0)\right).\label{eq:eq2}
\end{align}
Considering data with $\varphi\equiv 0$, we see from the $\elb$-independence of the left-hand
side of Eq.~\eqref{eq:eq2} that $\frac{1}{2}(R_\elb + R_{-\elb})=R$, independent of $\elb$. 
Then Eq.~\eqref{eq:eq1} becomes
\[
(\elb\cdot\nabla\varphi')(0) = \frac{1}{2}[R_\elb - R_{-\elb}] \pi(0)
\]
and, as $\pi(0)\in\CC^{|\nu|}$ is arbitrary, we deduce
by linearity of the left-hand side in $\elb$ that $\frac{1}{2}[R_\elb - R_{-\elb}]  = \ell^i \Delta_i$
for $(|\nu|\times|\nu|)$-matrices $\Delta_i$ ($1\le i\le n-1$). 
Substituting back into Eqs.~\eqref{eq:eq1} and~\eqref{eq:eq2}, and considering
data where $\pi\equiv 0$, we have $\pi'(0) = \ell^i\ell^j \Delta_i (\nabla_j \varphi)(0)$ 
for all unit vectors $\elb$. As the right-hand side is $\elb$-independent, we have
$\Delta_i (\nabla_j \varphi)(0) \propto \delta_{ij}$ for all $\varphi\in \CoinX{\RR^{n-1};\CC^{|\nu|}}$,
which is possible only if $\Delta_i=0$ for all $1\le i\le n-1$.  Thus 
$R_\elb=R$, independent of $\elb$, and it follows that $\pi'(0) = R\pi(0)$ and $(\nabla\varphi')(0) = R\varphi(0)$  for general data $(\varphi,\pi)$.
Further, because $S$ and hence $\tS$ commute with spatial translations, 
there is a fixed  unitary $R\in {\rm U}(|\nu|)$ so that
\[
\tS \left(\begin{array}{c} \varphi\\ \pi\end{array}\right)(x) = 
\left(\begin{array}{c} R\varphi(x) \\ R\pi(x)\end{array}\right).
\]

Finally, because $S$ also commutes with time translations we have $(S\phi)(t,x) = R\phi(t,x)$. But
both $\phi$ and $S\phi$ solve the same field equation
$\Box\varphi + \Mc^2\varphi = 0$, so $R$ commutes with $\Mc^2$ and  decomposes into block diagonal form,
$R=\bigoplus_{m\in\Mc} R_m$, where each $R_m\in {\rm U}(\nu(m))$. 
Since $S$ commutes with complex conjugation, the $R_m$ are real
and hence orthogonal.
\end{proof}

\section{Equivalent topologies on the automorphism group}
\label{appx:topology}

In this appendix, we study
various topologies that can be placed on the automorphism group
$\Aut(\langle\FF,\gamma,\Sf\rangle)$ of a theory
obeying assumptions (1)--(\ref{it:ReehSchlieder}) of Sec.~\ref{sect:compactness}.
The main purpose is to conclude the proof of Theorem~\ref{thm:end_aut}, but
some of the arguments may be of independent use.
The topologies concerned are:
\begin{tightitemize}
\item the {\em $\Mb$-topology}, defined as the weakest in which
the function $\eta\mapsto \omega(\eta_\Mb(A))$ is
continuous on $G$, for every $A\in\Ff(\Mb)$, $\omega\in\Sf(\Mb)$; 
\item the {\em local $\Mb$-topology}, defined as above, but
restricting to $A$ belonging to kinematic local algebras of 
$\FF^\kin(\Mb;D)$ of truncated multi-diamonds $D$ in $\Mb$;
\item the {\em $\omega$-topology} induced by any gauge-invariant state $\omega$ inducing a faithful GNS representation (see the discussion following Proposition~\ref{prop:urep});
\item the {\em diamond topology}, namely the
join of all $\Mb$-topologies as $\Mb$ runs over truncated multi-diamond spacetimes. 
\item the {\em natural weak topology} (as defined in Sec.~\ref{sec:gauge_gp}), namely the join of all $\Mb$-topologies as $\Mb$ runs through $\Loc$.
\end{tightitemize}

We need a technical result. Suppose $G$ is a topological group acting (not necessarily continuously) by
automorphisms on a unital $C^*$-algebra $\Ac$. For any
state $\omega$ on $\Ac$, we say that 
an element $B\in\Ac$ is $(\omega,G)$-continuous if 
$\eta\mapsto F_{A,B}(\eta):=\omega(A^*\eta(B)A)$ is continuous on $G$ for every
$A\in\Ac$ with $\omega(A^*A)\neq 0$. The set of all $(\omega,G)$-continuous elements of $\Ac$ will be denoted $\Cc_{\omega,G}(\Ac)$; any subalgebra
of $\Ac$ contained in $\Cc_{\omega,G}(\Ac)$ will be described as
$(\omega,G)$-continuous.
\begin{Prop} 
With the preceding definitions and notation,
\begin{tightenumerate}\renewcommand{\theenumi}{(\roman{enumi})}
\item $\Cc_{\omega,G}(\Ac)$ is a self-adjoint,
norm-closed, linear subspace of $\Ac$;
\item if $B$ is $(\omega,G)$-continuous then $\eta\mapsto \omega(P\eta(B)Q)$ is continuous for every $P,Q\in\Ac$;
\item if $B$ and $C$, together with at least one of $BB^*$ or $C^*C$,
are $(\omega,G)$-continuous, then so is $BC$;
\item if $\Bc_\alpha$ are $(\omega,G)$-continuous subalgebras of $\Ac$, then $\bigvee_\alpha \Bc_\alpha$ is $(\omega,G)$-continuous;
\item if $\omega_n\to \omega$ in the uniform topology on $\Ac^*_{+,1}$ 
then $\displaystyle\cl \bigcap_n \Cc_{\omega_n,G}(\Ac) \subset \Cc_{\omega,G}(\Ac)$. 
\end{tightenumerate}
\end{Prop}
\begin{proof} (i). Linearity and self-adjointness are obvious. As $\Ac$
is a $C^*$-algebra, we may estimate $|\omega(A^*\eta(B)A)|\le \omega(A^*A)\|B\|$
and deduce that $B_n\to B$ in $\Ac$ implies that $F_{A,B_n}$ converges uniformly to
$F_{A,B}$, which is therefore continuous; thus  $\Cc_{\omega,G}(\Ac)$ is norm-closed. 
(ii). By polarisation, it is enough to show that $B\in \Cc_{\omega,G}(\Ac)$ implies
that $F_{P,B}$ is continuous for every $P\in \Ac$ (regardless of whether $\omega(P^*P)\neq 0$); the latter is seen to hold on noting that
$F_{\II+\lambda P,B}$
must be continuous for all sufficiently small $\lambda\in\CC$, and 
consequently each term in the expansion of
$F_{\II+\lambda P,B}$ in $\lambda,\bar{\lambda}$ must also be continuous. 
(iii). Let $\eta_0\in G$ be arbitrary and assume without loss that $BB^*\in \Cc_{\omega,G}(\Ac)$.\footnote{If $C^*C\in \Cc_{\omega,G}(\Ac)$, 
we may use the following argument to conclude that $C^*B^*\in \Cc_{\omega,G}(\Ac)$
and the required result follows on taking adjoints.}
The identity $\eta(BC)=\eta_0(B)\eta(C) + (\eta(B)-\eta_0(B))\eta(C)$ entails 
\[
F_{A,BC}(\eta) = \omega(A^*\eta_0(B)\eta(C)A) + 
\omega(A^*(\eta(B)-\eta_0(B))\eta(C) A),
\]
the first term of which is continuous in $\eta$ by part (ii). It suffices to
show that the last term vanishes in the limit $\eta\to\eta_0$.
To this end, the Cauchy--Schwarz inequality gives, 
\begin{align*}
|\omega(A^*(\eta(B)-\eta_0(B))\eta(C) A)|^2
&\le \omega(A^*(\eta(B)-\eta_0(B))(\eta(B^*)-\eta_0(B^*))A)
\omega(A^*\eta(C^*C)A) \\
&\le \|A\|^2\|C\|^2 \omega(A^*(\eta(B)-\eta_0(B))(\eta(B^*)-\eta_0(B^*))A).
\end{align*}
Expanding the right-hand side, we obtain a sum of functions known to be continuous in $\eta$, using $BB^*\in \Cc_{\omega,G}(\Ac)$ and part~(ii); moreover the expression vanishes for $\eta=\eta_0$. Thus $BC\in \Cc_{\omega,G}(\Ac)$. 
Part (iv) is now immediate, using~(iii) together with norm-closure and linearity of
$\Cc_{\omega,G}(\Ac)$. (v) If $B$ is $(\omega_n,G)$-continuous
for each $n$, then the estimate 
$|\omega(A^*\eta(B)A)-\omega_n(A^*\eta(B)A)|\le \|\omega-\omega_n\|
\|A\|^2\|B\|$ shows that each $F_{A,B}$ is the uniform limit of 
continuous functions, so $B\in \Cc_{\omega,G}(\Ac)$; this
establishes $\bigcap_n \Cc_{\omega_n,G}(\Ac) \subset \Cc_{\omega,G}(\Ac)$ and we take closures to complete the proof. 
\end{proof}

As a first application of this result, let $G=\Aut(\langle \Ff,\gamma,\Sf\rangle)$, endowed with the local $\Mb$-topology for
some $\Mb$, and let $\omega\in\Sf(\Mb)$. Then for every truncated
truncated multi-diamond $D$ in $\Mb$,
$\Ff^\kin(\Mb;D)$ is $(\omega,G)$-continuous; as these subalgebras generate $\Ff(\Mb)$, the whole algebra is $(\omega,G)$-continuous by part~(iv). Letting $\omega$ vary in $\Sf(\Mb)$, we see that every function generating the $\Mb$-topology is continuous in the local $\Mb$-topology; as the latter is trivially weaker than the former, we
have proved:
\begin{Prop} \label{prop:localtop}
The $\Mb$-topology coincides with the
local $\Mb$-topology.
\end{Prop}
For a second application, suppose $\omega_0\in\Sf(\Mb)$ is
a gauge-invariant state inducing a faithful GNS representation
of $\Ff(\Mb)$, and endow $G$ with the $\omega_0$-topology induced
by the strong operator topology in this representation. 
Hence the corresponding unitary representation of $G$ is strongly continuous and it follows that
$\Ff(\Mb)$ is $(\omega,G)$-continuous with respect to
any vector state $\omega$ in the GNS representation;
one sees by part (v) that this is also true for any $\omega$ in the
folium $\Fol(\omega_0)$ of $\omega_0$. Now consider any $\omega\in\Sf(\Mb)$
and a truncated multi-diamond $D$ in $\Mb$. By local quasi-equivalence, 
there is a state $\omega'\in\Fol(\omega_0)$ that
agrees with $\omega$ on $\Ff^\kin(\Mb;D)$, which is
$(\omega',G)$-continuous, and hence $(\omega,G)$-continuous. 
Applying part~(iv) again, additivity entails that $\Ff(\Mb)$ is
$(\omega,G)$-continuous for every $\omega\in\Sf(\Mb)$. Summarising:
\begin{Prop} \label{prop:omegatop}
Suppose $\omega_0\in\Sf(\Mb)$ is
a gauge-invariant state inducing a faithful GNS representation
of $\Ff(\Mb)$. Then the $\omega_0$-topology coincides with the $\Mb$-topology, and hence the $\Mb$-local topology.
\end{Prop}

We now specialise to Minkowski space, $\Mb_0$, and the Minkowski
vacuum state $\omega_0$, with GNS representation $(\HH_0,\pi_0,\Omega_0)$. Our aim is to show that the $\Mb_0$- and $\omega_0$-topologies, which coincide by Proposition~\ref{prop:omegatop}, are equivalent to the diamond topology. 

Let $D$ be any multi-diamond of $\Mb_0$, and let $\iota:\Db\to \Mb_0$ be the canonical inclusion of $\Db=\Mb_0|_D$ in $\Mb_0$. Then $\Sf(\iota)\omega_0\in\Sf(\Db)$ is gauge-invariant and induces a GNS representation $(\HH_\Db,\pi_\Db,\Omega_\Db)$ of $\Ff(\Db)$
on which $G$ is unitarily represented. As $\omega_0$ has the Reeh--Schlieder property, we may take $\HH_\Db=\HH_0$, $\Omega_\Db=\Omega_0$, $\pi_\Db=\pi_0|_{\Ff^\kin(\Mb;D)}$, 
whereupon the unitary implementations of $G$ in the two representations also coincide. It follows (a) that the topologies induced on $G$ by $\Sf(\iota)\omega_0$ and $\omega_0$ are equivalent, and (b)
that $\pi_\Db$ is faithful, so the $\Sf(\iota)\omega_0$-topology and $\Db$-topologies coincide by Proposition~\ref{prop:omegatop}.  
Thus the $\Db$-topology coincides with the $\omega_0$-topology and hence the $\Mb_0$-topology. Using the time-slice property, we conclude
immediately that:
\begin{Prop} \label{prop:diamondtop}
The $\Mb_0$-topology coincides with the $\Db$-topology for
every truncated multi-diamond spacetime $\Db$, and hence with the
diamond topology. 
\end{Prop}

Our final observation is that the local $\Mb$-topology is trivially
weaker than the diamond topology for any $\Mb$. Combining this with Props.~\ref{prop:diamondtop} and~\ref{prop:localtop}, 
we find that the $\Mb_0$-topology is stronger than every $\Mb$-topology. 
Using the definition of the natural weak topology, we obtain in conclusion:
\begin{Thm} \label{thm:tops}
For a theory $\langle\FF,\gamma,\Sf\rangle\in\LCT_\grCAS$ obeying assumptions (1)--(\ref{it:ReehSchlieder}) of Sec.~\ref{sect:compactness}, 
the natural weak topology on $\Aut(\langle\FF,\gamma,\Sf\rangle)$ coincides with the $\omega_0$-topology (and hence with the diamond- and $\Mb_0$-topologies). 
\end{Thm}

\end{document}